\pdfoutput=1

\newif\iffull
\fulltrue
\iffull
\documentclass[sigplan,Proceedings]{acmart}
\patchcmd{\maketitle}{\@copyrightspace}{}{}{}
\renewcommand\footnotetextcopyrightpermission[1]{} % removes footnote with conference information in first column
\acmConference[]{}{}{}
\else
\documentclass[sigplan,screen]{acmart}
\fi
\iffull

\setcopyright{none}
\else
\acmPrice{15.00}
\acmDOI{10.1145/3385412.3385979}
\acmYear{2020}
\copyrightyear{2020}
\acmSubmissionID{pldi20main-p106-p}
\acmISBN{978-1-4503-7613-6/20/06}
\acmConference[PLDI '20]{Proceedings of the 41st ACM SIGPLAN International Conference on Programming Language Design and Implementation}{June 15--20, 2020}{London, UK}
\acmBooktitle{Proceedings of the 41st ACM SIGPLAN International Conference on Programming Language Design and Implementation (PLDI '20), June 15--20, 2020, London, UK}

\setcopyright{acmlicensed}
\fi
\startPage{1}

\usepackage{booktabs}   %% For formal tables:
\usepackage{subcaption}
\usepackage{mathtools}
\usepackage[leftbars,color]{changebar}
\nochangebars
\cbcolor{red}

\usepackage{xspace}
\usepackage{listings}
\usepackage{amsmath}
\usepackage{lipsum}
\usepackage{pifont}
\usepackage{amssymb}
\usepackage{color}
\usepackage{relsize}
\usepackage{wrapfig}
\usepackage{float}
\usepackage{calc}
\usepackage{hyperref}
\usepackage{refs}
\usepackage{graphicx}
\usepackage{pgfplots}
\usepackage{verbatim,bm}
\usepackage{stmaryrd}
\usepackage{multirow}
\usepackage{array}
\usepackage{tikz}
\usepackage{tabularx}
\usepackage{arydshln}
\usetikzlibrary{arrows,positioning}
\usepackage[T1]{fontenc}
\usepackage{etoolbox}
\usepackage{multicol}
\usepackage[noline,noend,ruled,linesnumbered]{algorithm2e} 
\makeatletter
\g@addto@macro\normalsize{%
	\setlength\abovedisplayskip{4pt}
	\setlength\belowdisplayskip{3pt}
	\setlength\abovedisplayshortskip{3pt}
	\setlength\belowdisplayshortskip{3pt}
}
\makeatother

\usepackage{enumitem}
\setlist[description]{
    labelindent=.1cm,
    style=unboxed,
    leftmargin=.3cm,
    font=\itshape,
    topsep=.5ex,
    itemsep=-.1ex
}

\setlist[itemize]{
    topsep=.5ex,
    leftmargin=.6cm,
    itemsep=-.1ex
}

\setlist[enumerate]{
    topsep=.5ex,
    leftmargin=.6cm,
    itemsep=-.1ex
}

\usepackage[breakable, theorems, skins]{tcolorbox}
\newenvironment{mybox}[1][gray!20]{
	\begin{tcolorbox}[   %% Adjust the following parameters at will.
		breakable,
		left=0pt,
		right=0pt,
		top=0pt,
		bottom=-1pt,
		colback=#1,
		colframe=#1,
		width=0.475\dimexpr\textwidth\relax,
		boxsep=4pt,
		arc=0pt,outer arc=0pt,
		]
	}{
\end{tcolorbox}
}

\newcounter{resq}%[section]
\newenvironment{resq}[1][]
{
	\vspace{-2mm}
	\refstepcounter{resq}\par\medskip
	\begin{mybox}
		\noindent \textbf{EQ~\theresq. #1} \rmfamily
	}{
\end{mybox}
\vspace{-2.5mm}
\medskip
}

\newcommand{\rone}{(\emph{i})~}
\newcommand{\rtwo}{(\emph{ii})~}

\definecolor{mGreen}{rgb}{0,0.6,0}
\definecolor{mGray}{rgb}{0.5,0.5,0.5}
\definecolor{mPurple}{rgb}{0.58,0,0.82}

\lstdefinestyle{CStyle}{
	commentstyle=\color{mGreen},
	keywordstyle=\color{mGreen},
	numberstyle=\tiny\color{mGray},
	escapeinside={(*@}{@*)},
	stringstyle=\color{mPurple},
	basicstyle=\footnotesize\ttfamily,
	breakatwhitespace=false,
	breaklines=true,
	captionpos=b,
	keepspaces=true,
	numbers=left,
	numbersep=5pt,
	showspaces=false,
	showstringspaces=false,
	showtabs=false,
	tabsize=2,
	language=C
}

\newcommand{\mypar}[1]{\vspace{0.5mm}\noindent  \textit{#1}}

\newcommand{\sem}[1]{\llbracket #1 \rrbracket\xspace}
\newcommand{\semgamma}[1]{\llbracket #1 \rrbracket^\#\xspace}

\newcommand\eqdef{\stackrel{{\normalfont\mbox{\tiny{def}}}}{=}}

\newcommand{\xmark}{\ding{55}}%

\newcommand{\mkadditive}{h}
\newcommand{\Se}{\mathcal S}

\newcommand{\slint}{\mathcal{S\!L}}

\newcommand{\linset}[1]{\langle #1\rangle}

\newcommand{\project}{\textsc{proj}_{\slint}}
\newcommand{\projectls}{\textsc{proj}_{\Se}}
\newcommand{\projectint}{\textsc{proj}_{\vec{\integ}}}
\newcommand{\removeif}{\textsc{RemIf}}

\def\tr{\textsf{t}}
\def\fa{\textsf{f}}

\def\var{\mathcal{V}}
\newcommand{\bset}{\textit{bset}}
\newcommand{\eqs}{\textit{eqs}}

\def\name{\textsc{nay}\xspace}
\def\namesl{\textsc{nay}$_\slint$\xspace}
\def\namehorn{\textsc{nay}$_\text{Horn}$\xspace}
\def\nope{\textsc{nope}\xspace}
\def\sygus{\textsc{SyGuS}\xspace}

\def\seahorn{\textsc{SeaHorn}\xspace}
\def\z3{\textsc{Z3}\xspace}
\def\cvc4{\textsc{CVC4}\xspace}
\def\cex{{E}\xspace}
\def\cegis{\textsc{CEGIS}\xspace}

\def\solved{70\xspace}
\def\solvednope{59\xspace}
\def\naymore{11\xspace}

\def\numSolved{\solved}
\def\numBenchmarks{132\xspace}
\def\numIfBenchmarks{57\xspace}
\def\numPlusBenchmarks{30\xspace}
\def\numConstBenchmarks{45\xspace}

\def\avgtime{1.97s\xspace}
\def\avgtimenope{15.59s\xspace}
\def\avgtimehorn{0.63s\xspace}

\def\nat{\mathbb{N}}
\def\integ{\mathbb{Z}}
\def\bool{\mathbb{B}}
\def\sy{{sy}}

\def\node{q}

\newcommand{\gammaHat}{\widehat{\gamma}}
\newcommand{\LogicLang}{\mathcal{L}}

\newcommand{\eval}[1]{{\lsyn #1 \rsyn}}
\newcommand{\lsyn}{\lbrack\!\lbrack}
\newcommand{\rsyn}{\rbrack\!\rbrack}
\newcommand{\B}[1]{\langle{#1}\rangle}

\newcommand{\calG}{\mathcal{G}}
\newcommand{\calGclia}{{{\mathcal{G}^{\tiny\text{CLIA}+}_{\tiny E}}}}
\newcommand{\calGOne}{{\mathcal{G}_1}}
\newcommand{\calD}{\mathcal{D}}
\newcommand{\calP}{\mathcal{P}}

\newcommand{\LIA}{\text{LIA}\xspace}
\newcommand{\LIAPLUS}{{\text{LIA}^+}}

\newcommand{\CLIAPLUS}{{\text{CLIA}^+}}
\newcommand{\Start}{\textit{Start}\xspace}
\newcommand{\SOne}{\textit{S1}}
\newcommand{\STwo}{\textit{S2}}
\newcommand{\SThree}{\textit{S3}}

\DeclareMathOperator{\extend}{\otimes}
\DeclareMathOperator{\combine}{\oplus}

\DeclareMathOperator*{\Combine}{\bigoplus}
\newcommand\wideDownarrow{\mathrel{\scalebox{1.5}[1]{$\Downarrow$}}}
\newcommand\bigwideDownarrow{\mathrel{\scalebox{1.5}[1]{$\bigg\Downarrow$}}}
\newcommand{\ostar}{\circledast}
\newcommand{\srZero}{\underline{0}}
\newcommand{\srOne}{\underline{1}}
\newcommand{\wsum}{\displaystyle\combine\limits_{i \in \nat}}
\newcommand{\Forall}[1]{\forall {#1}\,.}
\newcommand{\Exists}[1]{\exists {#1}\,.}

\newcommand{\Omit}[1]{}

\def\lang{L}

\newcommand{\ballns}[1]{\tikz[baseline=(myanchor.base)] \node[scale=0.9,circle,fill=.,inner sep=1pt,minimum size=1pt] (myanchor) {\color{-.}\bfseries\footnotesize #1};}

\begin{document}

%% Title information
\title[Proving Unrealizability of \sygus Problems]{Exact and Approximate Methods for Proving Unrealizability of Syntax-Guided Synthesis Problems}         %% [Short Title] is optional;
                                        %% when present, will be used in
                                        %% header instead of Full Title.
%\titlenote{with title note}             %% \titlenote is optional;
                                        %% can be repeated if necessary;
                                        %% contents suppressed with 'anonymous'
%\subtitle{Subtitle}                     %% \subtitle is optional
%\subtitlenote{with subtitle note}       %% \subtitlenote is optional;
                                        %% can be repeated if necessary;
                                        %% contents suppressed with 'anonymous'

%% Author information
%% Contents and number of authors suppressed with 'anonymous'.
%% Each author should be introduced by \author, followed by
%% \authornote (optional), \orcid (optional), \affiliation, and
%% \email.
%% An author may have multiple affiliations and/or emails; repeat the
%% appropriate command.
%% Many elements are not rendered, but should be provided for metadata
%% extraction tools.

%% Author with single affiliation.
\author{Qinheping Hu}
\affiliation{            %% \department is recommended
  \institution{University of Wisconsin-Madison}    \country{USA}         %% \country is recommended
}   %% \email is recommended
\author{John Cyphert}
\affiliation{            %% \department is recommended
	\institution{University of Wisconsin-Madison}  \country{USA}           %% \country is recommended
}   %% \email is recommended
\author{Loris D'Antoni}
\affiliation{            %% \department is recommended
	\institution{University of Wisconsin-Madison}  \country{USA}           %% \country is recommended
}   %% \email is recommended
\author{Thomas Reps}
\affiliation{            %% \department is recommended
	\institution{University of Wisconsin-Madison}  
	\country{USA}        %% \country is recommended
}   %% \email is recommended

%% Abstract
%% Note: \begin{abstract}...\end{abstract} environment must come
%% before \maketitle command
\begin{abstract}
We consider the problem of automatically establishing that a
given syntax-guided-synthesis (\sygus) problem is
unrealizable (i.e., has no solution). 
We formulate the problem of proving that a \sygus
problem is unrealizable over a finite set of examples as one of
solving a set of equations:
the solution yields an overapproximation of the set of possible outputs
that any term in the search space can produce on the given examples.
	If none of the possible outputs agrees with all of the examples,
our technique has proven that the given \sygus problem is unrealizable.
We then present an algorithm for exactly solving the set of equations that result from \sygus problems over 
linear integer arithmetic (LIA) and  LIA with conditionals (CLIA), thereby
showing that LIA and CLIA \sygus problems over finitely many examples are decidable.
We implement the proposed technique and algorithms in a tool called \name.
\name can prove unrealizability for \solved/\numBenchmarks  existing \sygus benchmarks,
with running times comparable to those of the state-of-the-art tool \nope. Moreover, \name can solve \naymore benchmarks that
\nope cannot solve.
\end{abstract}

%% 2012 ACM Computing Classification System (CSS) concepts
%% Generate at 'http://dl.acm.org/ccs/ccs.cfm'.

%% End of generated code

%% Keywords
%% comma separated list
  %% \keywords are mandatory in final camera-ready submission

%% \maketitle
%% Note: \maketitle command must come after title commands, author
%% commands, abstract environment, Computing Classification System
%% environment and commands, and keywords command.
\maketitle

% -*- TeX-master: t; TeX-PDF-mode: t -*-

\section{Introduction}
\label{Se:Introduction}

The goal of program synthesis is to find a program in some search
space that meets a specification---e.g., satisfies a set of examples or a
logical formula.
Recently, a large family of synthesis problems has been unified into a framework called
\emph{syntax-guided synthesis} (\sygus).
A \sygus problem is specified by a regular-tree grammar that describes
the search space of programs,
and a logical formula that constitutes the behavioral specification.
Many synthesizers support a specific format for \sygus problems~\cite{sygus},
and compete in annual synthesis competitions~\cite{alur2016sygus}.
These solvers are now quite mature and
are finding a wealth of applications~\cite{Hasan16,pldi17inversion}.

While existing \sygus synthesizers are good at finding a solution 
when one exists, there has been only a small amount of work on methods to prove that
a given \sygus  problem is \emph{unrealizable}---i.e., the problem does not admit a solution.
The problem of proving unrealizability arises in applications such as pruning 
infeasible paths in symbolic-execution engines~\cite{Mechtaev18}
and computing syntactically optimal solutions to \sygus problems~\cite{HuD18}. 
However, proving that a \sygus problem is unrealizable is particularly hard and, in general,
undecidable~\cite{CaulfieldRST15}. 
When a \sygus problem is realizable, any search technique that
	systematically explores the infinite search space of possible programs
	will eventually identify a solution to the synthesis problem. 
	In contrast, proving that a problem is unrealizable requires showing
	that \emph{every} program in the \emph{infinite} search space
	\emph{fails to
		satisfy} the specification.

Although we cannot hope to have a complete algorithm for establishing unrealizability, 
the goal of this paper is to develop a framework for solving 
the kinds of problems that appear in practice.
Our framework can be used in tandem with existing synthesizers that use the 
\emph{counterexample-guided inductive synthesis} (CEGIS) approach, in which the synthesizer iteratively builds 
a set of input examples  and finds programs consistent with the examples.
%If the problem being tackled by the synthesizer is unrealizable, the
%CEGIS algorithm might eventually find a set of examples for which no
%solution exists in the input grammar, at which point most synthesizers
%would get stuck in an infinite search over the input grammar.
%At each step, a tool implementing
%our framework would be executed in parallel to attempt to prove that no solution exists with the current set of examples. If the tool succeeds, the answer ``unrealizable'' is returned.

Our approach builds on the observation that unrealizability of a
\sygus problem $\sy$ can be proved by showing, for some finite
set of examples $\cex$, that $\sy^\cex$---the same problem with the
weaker specification of merely satisfying the examples in
$\cex$---is unrealizable \cite{cav19}.
We combine this observation with techniques from the abstract-interpretation literature to show that determining realizability of a linear integer arithmetic (LIA) \sygus problem over a finite set of examples  is actually \emph{decidable}. 
Our work gives a decision procedure to show unrealizability for a $\sy^E$ instance, whereas the prior work by \citet{cav19} reduced the problem to a program-reachability problem. In their approach, if an assertion inside a constructed program is shown to be valid, then the original problem is unrealizable. The issue with prior work is that the resulting reachability problem is passed to an incomplete solver that may not terminate or may only return unknown.

Even though we consider a finite set of examples, showing
realizability is non-trivial because the grammar
can still generate an infinite set of terms.
The main idea of this paper is to use an abstract domain to
overapproximate the possibly infinite set of outputs that
the terms derivable from each non-terminal of the grammar
of $\sy^\cex$ can produce on examples $\cex$.
The overapproximation is formalized using \emph{grammar-flow-analysis} (GFA),
a method that extends dataflow analysis to grammars \cite{SAGA:MW91}.
We define a GFA problem whose solution associates an overapproximating
abstract-domain value with each non-terminal of the \sygus grammar.
We then use the notion of \emph{symbolic concretization}
\cite{VMCAI:RSY04} to represent the abstract values as logical
formulas, which get combined with the \sygus specification to
produce an SMT query whose result can imply that the original problem
is unrealizable.

Using this framework, a variety of abstract domains can be used to
show unrealizability for arbitrary \sygus problems.
However, we also give a particular instantiation of the framework to
obtain a \emph{decision procedure for (un)realizability of LIA \sygus problems
over a finite set of examples}.
The key to this reduction is the use of the abstract domain of
\emph{semi-linear sets}.
We show that the GFA problem over semi-linear sets can be solved to
yield a semi-linear set that \emph{exactly} captures the set of
 possible outputs of the \sygus grammar.
	The problem $\sy^E$ is unrealizable if and only if the semi-linear set
	for the start non-terminal of the  grammar contains no value that satisfies the specification.
We extend this result to \sygus problems whose grammar contains LIA
terms and conditionals (CLIA).
%% Our decidability result carries over to this setting as well.

Our work makes the following three contributions:

\noindent \textbf{(1)} We reduce the problem of proving unrealizability
    of a \sygus problem, where the specification is given by examples,
    to the problem of solving a set of equations
    in an abstract domain (\sectref{overview}). 
    The correctness of our reduction is based on the
    framework of grammar-flow analysis (\sectref{background} and
    \sectref{gfa}).

\noindent \textbf{(2)} We show that the equations resulting from our reduction can be solved exactly for
    \sygus problems in which the grammars only generate terms in LIA
    (\sectref{ProvingUnrealizabilityOfSyGuSProblemsInLIA}) 
    and CLIA (\sectref{SolvingSyGuSProblemsInCLIA}), therefore yielding
    the first \emph{decision procedures} for LIA and CLIA 
    \sygus problems over a finite set of examples.

\noindent \textbf{(3)} We implement our technique in a tool, \name  (\sectref{implementation}).
    \name can prove unrealizability for \numSolved/\numBenchmarks benchamrks that were used to evaluate the state-of-the-art tool
    \nope. In particular, \name can solve \naymore benchmarks that \nope could not solve  (\sectref{evaluation}).

\sectref{RelatedWork} discusses related work.
\iffull
\else
Proofs and additional details can be found in the supplementary material.
\fi

%%% Local Variables: 
%%% mode: latex
%%% TeX-master: "main.tex"
%%% End: 

% -*- TeX-master: t; TeX-PDF-mode: t -*-

\section{Illustrative Examples}
\label{Se:overview}
\mypar{\sygus problems in LIA.} Consider the \sygus problem in which the goal is to create a term $e_f$ whose meaning is 
$e_f(x):=2x+2$, but where $e_f$ is in the language of the following regular tree grammar $G_1$:\footnote{
\label{Footnote:ExpandedGrammar}
For readability, we allow grammars to contain $n$-ary Plus symbols and trees. 
In the next sections, we will
write the grammar $G_1$ as follows:
$\arraycolsep=1.4pt
\begin{array}{rclrcl}
\hspace{5mm}\Start	 &::=&  \textrm{Plus}(\SOne, \Start) \mid \textrm{Num}(0) & \hspace{5mm}
\SOne	 &::=&  \textrm{Plus}(\STwo, \textrm{Var}(x))\\
\STwo	 &::=& \textrm{Plus}(\SThree, \textrm{Var}(x)) &
\SThree &::=& \textrm{Var}(x).
\end{array}$
}
\begin{align}
\Start	::= \textrm{Plus}(\textrm{Var}(x), \textrm{Var}(x), \textrm{Var}(x), \Start) \mid \textrm{Num}(0) \label{Eq:grammar1}	
\end{align}
This problem is unrealizable because every term in the grammar $G_1$ is of the form $3kx$ (with $k\geq 0$).
%In the following, we show how our technique can prove this fact automatically.

A typical synthesizer tries to solve this problem using  a counterexample-guided inductive synthesis (CEGIS)
strategy
that searches for a program consistent with a finite set of examples $E$.
Here, let's assume that the initial input example in $E$ is $i_1$, which has $x$ set to $1$---i.e $i_1(x)=1$.
For this example, the input $i_1$ corresponds to the output $o_1=4$. 

In this particular case, there exists no term in the grammar $G_1$ that is  consistent with the example $i_1$.
To prove that this grammar does not contain a term that is consistent
with the specification on the example $i_1$, we compute for each
nonterminal $A$ a value $n_{1,E}(A)$ \footnote{
This section uses a simplified notation for readability. In \sectref{gfa} 
the term $n_{1,E}(A)$ is written $n_{\calG_{1E}}$ where $\calG_1$ is used to denote
a GFA problem.} that describes the set of values any term
derived from $A$ can produce when evaluated on $i_1$---i.e.,
$\gamma({n_{1,E}(A)}) \supseteq \{\sem{e}(i_1)\mid e\in \lang_{G_1}(A)\}$,
where, as usual in abstract interpretation, $\gamma$ denotes the concretization function.
As we  show in \sectref{gfa},  for $n_{1,E}(A)$ to be an overapproximation of the set of output values that
any term derived from $A$ can produce for the current set of examples $E$, it should satisfy the following equation:
\begin{equation}
\begin{split}
n_{1,E}(\Start)	& =	 
\semgamma{\textrm{Plus}}_E(\semgamma{\textrm{Var}(x)}_E,\semgamma{\textrm{Var}(x)}_E,\semgamma{\textrm{Var}(x)}_E,\\
& n_{1,E}(\Start)) 
\oplus \semgamma{\textrm{Num}(0)}_E.\label{Eq:gfa1}	
\end{split}
\end{equation}
For every term $e$, the notation $\semgamma{e}_E$ denotes an abstract semantics of $e$---i.e.,
$\semgamma{e}_E$ overapproximates the set of values $e$ can produce when evaluated on the examples in $E$---and $\oplus$ denotes the \emph{join} operator, which overapproximates $\cup$.

In this example, we represent each $n_{1,E}(A)$ using a
\emph{semi-linear set}---i.e., a set of terms $\{l_1, \ldots, l_n\}$,
where each $l_i$ is a term of the form $c+\lambda_1c_1+\cdots+\lambda_kc_k$
(called a \emph{linear set}), the values $\lambda_i \in \nat$ are parameters,
and the values $c_j \in \integ$ are fixed coefficients.
We then replace each $\semgamma{e}_E$ with a corresponding semi-linear-set interpretation. 
For example, $\semgamma{\textrm{Var}(x)}_E$ is the vector of inputs $E$
projected onto the $x$ coordinate---i.e., $\semgamma{\textrm{Var}(x)}_E= \{i_1(x)\} = \{1\}$.
We rewrite $\semgamma{\textrm{Plus}}_E$ as $\otimes$, with $x\otimes y$ being the semi-linear set representing $\{a+b\mid a\in x, b\in y\}$

We rewrite \eqref{gfa1} to use semi-linear sets:
\begin{equation}
  \label{Eq:gfa1subst}
  n_{1,E}(\Start) = \left( \{1\}\otimes \{1\} \otimes \{1\} \otimes n_{1,E}(\Start) \right) \oplus \{0\},
\end{equation}
where $x \oplus y$ is the semi-linear set representing $\{a\mid a\in x \vee a\in y\}$. 
These operations can be performed precisely.

In this example, an \emph{exact} solution to this set of equations is the semi-linear set
$n_{1,E}(\Start)=\{0+\lambda 3\}$, which describes the set of all possible  values 
produced by any term in  grammar $G_1$ for the set of examples 
$E=\B{i_1}$. In particular, such a solution can be computed automatically~\cite{EsparzaKL10}.\iffull{\footnote{
  Some intuition can be gained by thinking of \eqref{gfa1subst} as
  being similar to a context-free grammar of the form $X:=aX\mid b$,
  which has the regular-language solution $a^*b$.
  Similarly, \eqref{gfa1subst} has the solution $\{3\}^{\ostar} \otimes \{0\}$. 
  Here $\ostar$ is the iterated addition of the (trivial) semi-linear set $\{3\}$,
  so the overall solution is $\{0+3\lambda \mid \lambda\in\nat\}$.
}\fi
This \sygus problem does \emph{not} have a solution, because none of
the values in $n_{1,E}(\Start)$ meets the specification on the given input example, i.e., the following formula is not 
satisfiable:
\begin{align}
\exists \lambda. [i_1=1\wedge o_1=0+\lambda3 \wedge \lambda\geq 0] \wedge o_1=2 i_1+2. \label{Eq:unsat1}
\end{align}

\mypar{\sygus problems in CLIA.} For grammars with a more complex background theory, such as CLIA (LIA with conditionals), it may
be more complicated to compute an overapproximation of the possible outputs of any term in the grammar. 
For example, consider the \sygus problem where once again the goal is to synthesize a term whose meaning is $e_{f}(x):=2x + 2$, but now in the more expressive CLIA grammar $G_2$:
\begin{equation}
\label{Eq:grammar2}
{\small
\begin{array}{@{\hspace{0ex}}r@{~}r@{~~}l@{\hspace{0ex}}}
  \Start	& ::=	& \textrm{IfThenElse}(BExp, Exp3, \Start) \mid Exp2 \mid Exp3\\			
  BExp		& ::=	& \textrm{LessThan}(\textrm{Var}(x), \textrm{Num}(2))  \\
                & \mid  & \textrm{LessThan}(\textrm{Num}(0), \Start) \mid \textrm{And}(BExp,BExp) \\
  Exp2		& ::=	& \textrm{Plus}(\textrm{Var}(x),\textrm{Var}(x), Exp2) \mid \textrm{Num}(0)\\			
  Exp3		& ::=	& \textrm{Plus}(\textrm{Var}(x),\textrm{Var}(x),\textrm{Var}(x), Exp3) \mid \textrm{Num}(0)
\end{array}
}
\end{equation}
Consider again the input example $i_1{=}1$ with
output $o_1{=}4$. 
The
term $\textrm{Plus}(\textrm{Var}(x),\textrm{Var}(x), \textrm{Plus}(\textrm{Var}(x),\textrm{Var}(x), \textrm{Num}(0)))$
in this grammar is correct on the input $i_1$.
A \sygus solver that enumerates all terms in the grammar will find this term, test it on the given specification, see that it is not correct on all inputs, and produce a counterexample.
In this case, suppose that the counterexample is $i_2$ where $i_2(x){=}2$ with the corresponding output $o_2{=}6$.
There is no term in $G_2$ that is consistent with both of these examples, and 
we will prove this fact like we did before, that is, by solving the following set of equations:\footnote{
  Note that the $\oplus$ symbol is overloaded. On the right-hand side of
  $n_{2,E}(\textit{BExp})$, $\oplus$ is an operation on an
  abstract Boolean value, whereas the $\oplus$ on the right-hand-side of
  the other equations is an operation on semi-linear sets. Both
  operations denote set union, and are handled in a uniform way by
  operating over a multi-sorted domain of Booleans and semi-linear sets.
}
\begin{equation}
\label{Eq:gfa2}
{\small
\begin{array}{@{\hspace{0ex}}r@{~}r@{~~}l@{\hspace{0ex}}}
  {n_{2,E}(\Start)}	& = & \semgamma{\textrm{IfThenElse}}_E(n_{2,E}(BExp), n_{2,E}(Exp3), \\
                                &   & n_{2,E}(\Start)) \oplus n_{2,E}(Exp2) \oplus n_{2,E}(Exp3) \\
  n_{2,E}(BExp)         & = & \semgamma{\textrm{LessThan}}_E(\semgamma{\textrm{Var}(x)}_E, \semgamma{\textrm{Num}(2)}_E) \\
                           & \oplus & \semgamma{\textrm{LessThan}}_E(\semgamma{\textrm{Num}(0)}_E,n_{2,E}(\Start)) \\
                           & \oplus & \semgamma{\textrm{And}}_E(n_{2,E}(BExp),n_{2,E}(BExp)) \\
  n_{2,E}(Exp2)	        & = & \semgamma{\textrm{Plus}}_E(\semgamma{\textrm{Var}(x)}_E,\semgamma{\textrm{Var}(x)}_E, n_{2,E}(Exp2)) \\
                           & \oplus & \semgamma{\textrm{Num}(0)}_E \\
  n_{2,E}(Exp3)         & = & \semgamma{\textrm{Plus}}_E(\semgamma{\textrm{Var}(x)}_E,\semgamma{\textrm{Var}(x)}_E,\semgamma{\textrm{Var}(x)}_E, \\
                                &   & n_{2,E}(Exp3)) \oplus \semgamma{\textrm{Num}(0)}_E	
\end{array}
}
\end{equation}
Because we want to track the possible values each term can have for \emph{both} examples,
we need a domain that summarizes vectors of values. 
Luckily, semi-linear sets can easily be extended to vectors---i.e.,
 each $l_i$ in a semi-linear set $sl$ is a linear set of the form $\{\vec{v}_0+\lambda_1\vec{v}_1+\cdots+\lambda_k\vec{v}_k\mid \lambda_i\in\nat\}$ (with 
   $\vec{v}_j{\in} \integ^k$). 
Second, because some nonterminals are Boolean-valued and some are integer-valued, we need
different representations of the possible outputs of each nonterminal. 
We will use semi-linear sets for 
$n_{2,E}(\Start)$, $n_{2,E}(Exp2)$ and $n_{2,E}(Exp3)$, and 
a set of Boolean vectors for $n_{2,E}(BExp)$---e.g., $n_{2,E}(BExp)$
could be a set $\{(\tr,\fa),  (\tr,\tr)\}$, which denotes that  a Boolean expression generated by $BExp$  can
be true for $i_1$ and false for $i_2$, or true for both. 
\iffull
We can now instantiate all constant terminals and variable terminals with their abstractions, and rewrite the 
equations as follows:
\begin{equation}
\label{Eq:gfa2subst}
{\small
\begin{array}{@{\hspace{0ex}}r@{~}r@{~~}l@{\hspace{0ex}}}
  {n_{2,E}(\Start)}  & = & \semgamma{\textrm{IfThenElse}}_E(n_{2,E}(BExp), n_{2,E}(Exp3), \\
			  &   & n_{2,E}(\Start)) \oplus n_{2,E}(Exp2) \oplus n_{2,E}(Exp3) \\
  n_{2,E}(BExp)      & = & \{(\tr, \fa)\} \oplus \semgamma{\textrm{LessThan}}_E(\{(0,0)\},n_{2,E}(\Start)) \\
                     & \oplus & \semgamma{\textrm{And}}_E(n_{2,E}(BExp),n_{2,E}(BExp)) \\
  n_{2,E}(Exp2)      & = & [\{(1,2)\}\otimes\{(1,2)\}\otimes n_{2,E}(Exp2)] \oplus \{(0,0)\} \\
  n_{2,E}(Exp3)      & = & [\{(1,2)\}\otimes\{(1,2)\}\otimes\{(1,2)\}\otimes n_{2,E}(Exp3)] \\
                     & \oplus & \{(0,0)\}		
\end{array}
}
\end{equation}

\else
We can now instantiate all constant terminals and variable terminals with their abstraction, e.g., 
$\semgamma{\textrm{Var(x)}}_E$ with $\{(1,2)\}$ and $\semgamma{\textrm{Num}(0)}_E	$ with $\{(0,0)\}$.
\fi
We then start solving part of our equations by observing that $Exp2$ and $Exp3$ are only recursive in themselves. Therefore, we can compute
their summaries independently,
obtaining
$n_{2,E}(Exp2) = \{(0,0) + \lambda (2,4)\}, n_{2,E}(Exp3) =  \{(0,0)+\lambda (3,6)\}$.
We can now replace all instances of $n_{2,E}(Exp2)$ and $n_{2,E}(Exp3)$, and obtain the following set of equations:
\begin{equation}
\label{Eq:mutualrecursion}
{\small
\begin{array}{@{\hspace{0ex}}r@{~}r@{~~}l@{\hspace{0ex}}}
  {n_{2,E}(\Start)}    & = & \semgamma{\textrm{IfThenElse}}_E(n_{2,E}(BExp), \{(0,0)+\lambda (3,6)\}, \\
                            &   & n_{2,E}(\Start)) \oplus  \{(0,0)+\lambda (2,4)\} \\
                       & \oplus & \{(0,0)+\lambda (3,6)\} \\
  n_{2,E}(BExp)        & = & \{(\tr, \fa)\} \oplus \semgamma{\textrm{LessThan}}_E(\{(0,0)\},n_{2,E}(\Start)) \\
                       & \oplus & \semgamma{\textrm{And}}_E(n_{2,E}(BExp),n_{2,E}(BExp))
\end{array}
}
\end{equation}
We now have to face the problem of solving equations over  $n_{2,E}(BExp)$ and 
$n_{2,E}(\Start)$, which  represent different types of values and are mutually recursive.
Because  the domain of $n_{2,E}(BExp)$ is finite (it has at most $2^{|E|}$ elements),
we can solve the equations iteratively until we reach a fixed point for both variables.
In particular, we initialize all variables to the empty set and evaluate right-hand sides, so $n_{2,E}^{\texttt{0}}(BExp)=\{(\tr, \fa)\}$ (the superscript  denotes the iteration the algorithm is in).
We can replace $n_{2,E}(BExp)$ with the value of
$n_{2,E}^{\texttt{0}}(BExp)$ in the equation for  $n_{2,E}^{\texttt{1}}(\Start)$ as follows:
\begin{equation}
\label{Eq:iteonlyexpression}
{\small
\begin{array}{@{\hspace{0ex}}r@{~}r@{~~}l@{\hspace{0ex}}}
  {n_{2,E}^{\texttt{1}}(\Start)} & = & \semgamma{\textrm{IfThenElse}}_E(\{(\tr, \fa)\}, \{(0,0)+\lambda (3,6)\}, \\
                              &   & n_{2,E}^{\texttt{1}}(\Start))  \oplus  \{(0,0)+\lambda (2,4)\} \\
                         & \oplus & \{(0,0)+\lambda (3,6)\}	
\end{array}
}
\end{equation}
At this point, we face a new problem: we need to express the abstract semantics of IfThenElse using
the semi-linear set operators $\oplus$ and $\otimes$. In particular, 
we would like to produce a semi-linear set 
in which, for each vector,
some components come from the semi-linear set for the then-branch (i.e., values corresponding to inputs for which
the IfThenElse guard was true),
and some components come from the semi-linear set for the else-branch (i.e., values corresponding to inputs for which
the IfThenElse guard was false).
We overcome this problem by rewriting the above equations 
as follows:
\begin{equation}
\label{Eq:itesplitting}
{\small
\begin{array}{@{\hspace{0ex}}r@{~}r@{~~}l@{\hspace{0ex}}}
  {n_{2,E}^{\texttt{1}}(\Start^{(\tr,\tr)})}      & = & \{(0,0)+\lambda (3,0)\} \otimes n_{2,E}^{\texttt{1}}(\Start^{(\fa,\tr)}) \\
                                              & \oplus & \{(0,0)+\lambda (2,4)\} \oplus \{(0,0)+\lambda (3,6)\}\\			
  {n_{2,E}^{\texttt{1}}(\Start^{(\fa,\tr)})}      & = & \{(0,0)+\lambda (0,0)\} \otimes n_{2,E}^{\texttt{1}}(\Start^{(\fa,\tr)}) \\
                                              & \oplus & \{(0,0)+\lambda (0,4)\} \oplus \{(0,0)+\lambda (0,6)\}
\end{array}
}
\end{equation}
Intuitively, ${n_{2,E}^{\texttt{1}}(\Start^{(\fa,\tr)})}$  is the abstraction obtained by only executing the expressions generated by $\Start$ on the
second example and leaving the output of the first example as 0 to represent the fact that only 
the example $i_2$ followed the else branch of the IfThenElse statement.
Similarly, the semi-linear set $\{(0,0)+\lambda (3,0)\}$ zeroes out the second component of the semi-linear set appearing
in the then branch.
The value of $n_{2,E}^{\texttt{1}}(\Start^{(\tr,\tr)})$ (which is also the value of $n_{2,E}^{\texttt{1}}(\Start)$), is then computed by summing ($\otimes$) together the then and else values.
This set of equations is now in the form that we can solve automatically---i.e., it only involves the
 operations $\oplus$ and $\otimes$  over semi-linear sets---and thus we can compute  the value of  $n_{2,E}^{\texttt{1}}(\Start)$.
We  now plug that value into the equation for $BExp$  and compute the value of $n_{2,E}^{\texttt{1}}(BExp)$,
\begin{equation}
\label{Eq:bexpronly}
{\small
\begin{array}{@{\hspace{0ex}}r@{~}r@{~~}l@{\hspace{0ex}}}
  n_{2,E}^{\texttt{1}}(BExp)     & = & \{(\tr, \fa)\} \oplus \semgamma{\textrm{LessThan}}_E(\{(0,0)\},n_{2,E}^{\texttt{1}}(\Start)) \\
                         & \oplus & \semgamma{\textrm{And}}_E(n_{2,E}^{\texttt{1}}(BExp),n_{2,E}^{\texttt{1}}(BExp))		
\end{array}
}
\end{equation}
Because $n_{2,E}^{\texttt{1}}(BExp)$  has a finite domain, equations over such a domain can be solved iteratively, 
in this case yielding the fixed-point value 
$n_{2,E}^{\texttt{1}}(BExp)=\{(\tr, \fa),(\tr,\tr),(\fa,\fa)\}$.
We  now plug this solution into the equation for $\Start$ and compute 
the value of $n_{2,E}^{\texttt{2}}(\Start)$
similarly to how we computed that of $n_{2,E}^{\texttt{1}}(\Start)$.
We  then use $n_{2,E}^{\texttt{2}}(\Start)$ to compute $n_{2,E}^{\texttt{2}}(BExp)$
and discover that $n_{2,E}^{\texttt{2}}(BExp)=n_{2,E}^{\texttt{2}}(BExp)$.
	Because we have reached a fixed point, we have found the set of possible values the grammar can output on our set of examples, i.e., the abstraction $n_{2,E}^{\texttt{1}}(Start)$ captures all possible values the grammar $G_2$ can output on $E$.
	By plugging such values in the original formula similarly to what
	we did in \eqref{unsat1}  we get that no output set satisfies the formula on the given input examples,
	and therefore this \sygus problem is unrealizable.

%%% Local Variables: 
%%% mode: latex
%%% TeX-master: "main.tex"
%%% End: 

% -*- TeX-master: t; TeX-PDF-mode: t -*-

\section{Background}
\label{Se:background}
In this section, we recall the definition of syntax-guided synthesis
over a finite set of examples.

\subsection{Trees and Tree Grammars.}
A \emph{ranked alphabet} is a tuple $(\Sigma,rk_\Sigma)$ where
$\Sigma$ is a finite set of symbols and $rk_\Sigma:\Sigma\to\nat$
associates a rank to each symbol.
For every $m\ge 0$, the set of all symbols in $\Sigma$ with rank $m$
is denoted by $\Sigma^{(m)}$.
In our examples, a ranked alphabet is specified by showing the set
$\Sigma$ and attaching the respective rank to every symbol as a
superscript---e.g., $\Sigma=\{Plus^{(2)},Var(x)^{(0)}\}$.
(For brevity, the superscript is sometimes omitted.)
We use $T_\Sigma$ to denote the set of all (ranked) trees over
$\Sigma$---i.e., $T_\Sigma$ is the smallest set such that
\rone $\Sigma^{(0)} \subseteq T_\Sigma$,
\rtwo if $\sigma^{(k)} \in \Sigma^{(k)}$ and $t_1,\ldots,t_k\in T_\Sigma$, then
$\sigma^{(k)}(t_1,\cdots,t_k)\in T_\Sigma$.
In what follows, we assume a fixed ranked alphabet $(\Sigma,rk_\Sigma)$.

\begin{definition}[Regular-Tree Grammar]
A {\emph{regular tree grammar}} (RTG) is a tuple $G=(N,\Sigma,S,\delta)$,
where $N$ is a finite set of nonterminal symbols of arity 0;
$\Sigma$ is a ranked alphabet;
$S\in N$ is an initial nonterminal;
and $\delta$ is a finite set of productions of the form $A_0 \to \sigma^{(i)}(A_1, \ldots, A_i)$,
where for $1 \leq j \leq i$, each $A_j \in N$ is a nonterminal.
\end{definition}

Given a tree $t\in T_{\Sigma\cup N}$,
applying a production $r = A\to\beta$ to $t$ produces the tree $t'$ resulting from replacing
the left-most occurrence of $A$ in $t$ with the right-hand side $\beta$.
A tree $t\in T_\Sigma$ is generated by the grammar $G$---denoted by $t\in
L(G)$---iff it can be obtained by applying a sequence of
productions $r_1\cdots r_n$ to the tree whose root is the initial
nonterminal $S$. $\delta_A\subseteq\delta$ denotes the set of productions associated with nonterminal $A$, and $\Sigma_A:=\{\sigma^{(i)}\mid A\to\sigma^{(i)}(A_1,...,A_i)\in\delta_A\}$.

\subsection{Syntax-Guided Synthesis.}
A \sygus problem is specified with respect to a background theory $T$---e.g., linear arithmetic---and
the goal  is to synthesize a function $f$ that satisfies two constraints provided by the user. 
The first constraint, $\psi(f(\bar{x}),\bar{x})$, describes a \emph{semantic property} that $f$ should satisfy. 
The second constraint limits the \emph{search space} $S$ of $f$, and is
given as a set of terms specified by an RTG $G$
that defines a subset of all terms in $T$. 

\begin{definition}[\sygus]
A \sygus problem over a background theory $T$ is a pair 
$\sy =(\psi(f,\bar{x}), G)$, where
$G$ is a regular tree grammar that only contains terms in $T$---i.e., $L(G)\subseteq T$---and
$\psi(f,\bar{x})$ is a Boolean formula constraining the semantic behavior of the synthesized program $f$.\footnote{In this paper, we focus on \textit{single-invocation} \sygus problems for which the formula $\psi$
only contains instances of the function $f$ that are called on the input  $\bar{x}$.
We write $\psi(f,\bar{x})$ instead of $\psi(f(\bar{x}),\bar{x})$ for brevity.}

A \sygus problem is \textbf{\emph{realizable}} if there exists an expression $e\in L(G)$ such that
$\forall \bar{x}. \psi(\sem{e},\bar{x})$ is true. Otherwise we say that the problem is \textbf{\emph{unrealizable}}.
\end{definition}

\begin{theorem}[Undecidability~\cite{CaulfieldRST15}]
Given a \sygus problem $\sy$, it is undecidable to check whether $\sy$ is realizable.
\end{theorem}

Many \sygus solvers do not solve the problem of finding a term
that satisfies the specification on \emph{all} inputs.
Instead, they look for an expression that
satisfies the specification on a \emph{finite} example set $\cex$.
If such a term is found, it is then checked if it can be generalized to all inputs.
We take a similar approach to show unrealizability. 

\begin{definition}
Given a \sygus problem $\sy=(\psi(f,\bar{x}),G)$ and a finite set of inputs $\cex=\B{i_1,\ldots,i_n}$, let $\sy^\cex:=(\psi^\cex(f), G)$ denote the problem of finding a term $e\in L(G)$ such that $\sem{e}$ is only required to be correct on the examples in $\cex$.
Let $\sem{e}_\cex$ denote the vector of outputs $\B{\sem{e}(i_1),\ldots,\sem{e}(i_n)}$
($= \B{o_1,\ldots, o_n}$) produced by $e$ on $E$. A $\sy^\cex$ problem is \textbf{\emph{realizable}} if $\psi^\cex(\sem{e}_E) \eqdef \bigwedge_{i_j\in E} \psi(\sem{e}(i_j), i_j)$ holds, and \textbf{\emph{unrealizable}} otherwise.
\end{definition}

\begin{lemma}[\cite{cav19}]
\label{lem:soundness-cegis}
	If $\sy^\cex$ is unrealizable then $\sy$ is unrealizable.
\end{lemma}

\begin{example}
\label{Exa:LIA}
	The regular tree grammar of all \textit{linear integer arithmetic} (LIA) terms is
	\[
	T_{\LIA} ::= \textrm{Plus}(T_{\LIA},T_{\LIA}) \mid \textrm{Minus}(T_{\LIA},T_{\LIA}) \mid \textrm{Num}(c) \mid \textrm{Var}(x)
	\]
	where $c\in \integ$, and $x\in \var$ is an input variable to the function being synthesized. 	
	The semantics of these productions is as expected, and is extended to terms in the usual way.
	
	In the case of a $\sy^\cex$ instance, we consider the restricted semantics  of LIA with respect to a set of examples $E = \B{i_1, \ldots, i_n}$, given by a function
	$\sem{\cdot}_E \colon T_{LIA} \to \mathbb{Z}^n$.
        $\sem{\cdot}_E$ maps an LIA term to the corresponding output vector produced by evaluating the term with respect to
	all of the examples in $E$.
        Let $\mu_E \colon \var \to \mathbb{Z}^{n}$ be the function that projects the inputs onto the $x$ coordinate---i.e., $\mu_E(x) = \B{i_1(x), \ldots, i_n(x)}$.
	The semantics of the LIA operators with respect to an example set $E$ is then defined as follows:
	\[
	\begin{array}{r@{\hspace{1.0ex}}c@{\hspace{1.0ex}}l@{\hspace{5.0ex}}r@{\hspace{1.0ex}}c@{\hspace{1.0ex}}l}
	\sem{\textrm{Plus}}_E(\vec{v}_1,\vec{v}_2)  & := & \vec{v}_1 + \vec{v}_2   &   \sem{\textrm{Num}(c)}_E & := & \B{c,...,c} \\
	\sem{\textrm{Minus}}_E(\vec{v}_1,\vec{v}_2) & := & \vec{v}_1 - \vec{v}_2   &   \sem{\textrm{Var}(x)}_E & := & \mu_E(x)
	\end{array}
	\]
	where $+$ (resp.\ $-$) denotes the component-wise addition (resp.\ subtraction) of two vectors.
	$\sem{\cdot}_E \colon T_{\LIA} \to \mathbb{Z}^n$ is extended to terms in the usual way.			
For brevity, we overload the term ``$\LIA$''
 to refer both to the \emph{logic} $\LIA$ and to \LIA \emph{grammars}---i.e.,
 grammars over the alphabet $\{ \textrm{Plus}, \textrm{Minus}, \textrm{Num}(c), \textrm{Var}(x) \}$.
\end{example}

In \sectref{algorithm}, we present an algorithm based on Counterexample-Guided Inductive Synthesis (CEGIS) 
to show unrealizability of a \sygus problem, $\sy$, by showing unrealizability of a $\sy^\cex$ problem. 
The idea is to check unrealizability of $\sy^\cex$ for some set $\cex$. 
If $\sy^\cex$ is unrealizable, the algorithm reports unrealizable,
otherwise it generates a new example, $i_{n+1}$, adds it to $\cex'=\cex\cup \{i_{n+1}\}$, and tries
to prove unrealizability of $\sy^{\cex'}$, and so on.
In \sectref{ProvingUnrealizabilityOfSyGuSProblemsInLIA}, we show that the unrealizability problem for a $\sy^\cex$ instance is  decidable for LIA grammars.
 However, we note that there are 
\sygus problems for which 
\cegis-style algorithms \emph{cannot} prove unrealizability~\cite{cav19}.
%---i.e.,
%will it always terminate if $sy$ is unrealizable?
%The following lemma shows that \cegis algorithm is not powerful enough.
\iffull
\begin{lemma}[Incompleteness]
\label{lem:incompleteness-cegis}
	There exists an unrealizable \sygus problem $\sy$ such that for every finite set of 
	examples $\cex$ the problem $\sy^\cex$ is realizable.
\end{lemma}
The following example shows that CEGIS is incomplete for  \sygus problems over LIA grammars.
\begin{example}
Consider the \sygus problem $(\psi(f,x):=f(x)>x,G_{const})$ 
where $G_{const}$ is the following grammar that can produce for any constant value $c>0$ 
a term $e_c$ such that $\sem{e_c}=c$:
\[
\begin{array}{rcl}
		\Start &::=& \textrm{Plus}(\Start,\Start) \mid \textrm{Num}(1)
\end{array}
\]
For any finite set of examples $E$, we can find a constant term $e_{c'}$ in $L(G)$ whose semantics is $c'\geq \max(E)+1$, 
and therefore is a solution to $\sy^E$. Hence, a CEGIS algorithm cannot prove unrealizability for
this \sygus problem.\qed
\end{example}
\fi
Despite this negative result, we will show that a \cegis algorithm can prove unrealizability for many \sygus instances (\sectref{evaluation}).

%%% Local Variables: 
%%% mode: latex
%%% TeX-master: "main.tex"
%%% End: 

% -*- TeX-master: t; TeX-PDF-mode: t -*-

\section{Proving Unrealizability using Grammar Flow Analysis}
\label{Se:gfa}
In this section, we present a formalism called \emph{grammar flow analysis} (GFA) \cite{SAGA:MW91},
which connects regular tree grammars to equation systems, and show how to use GFA to prove
unrealizability of \sygus problems for finitely many examples.

\subsection{Grammar Flow Analysis}

GFA is a formalism used for equipping the language of a grammar with a
semantics \begin{changebar}
in which
%---actually \emph{two} semantics, which are defined in different ways, but are
%equivalent under conditions that often hold.
\end{changebar}
 the meaning of a tree is a value from a (complete) \emph{combine semilattice}.

\begin{definition}[Combine Semilattice]
  \label{def:semilattice}
  A \emph{combine semilattice} is an algebraic structure $\calD = (D,\oplus)$,
  where $\oplus: D \times D \rightarrow D$ is a binary operation on $D$ (called ``\emph{combine}'')
  that is commutative, associative, and idempotent.\iffull\footnote{
    We have chosen to use the neutral term ``combine,'' rather than
    meet or join, due to varying nomenclature in the literature.
    In our applications, if the semilattice is oriented according to the
    conventions of the abstract-interpretation literature, a combine-semilattice
    is a join-semilattice;
    if it is oriented according to the conventions of the dataflow-analysis
    literature, it is a meet-semilattice.
  }
  \begin{description}
    \item [Commutativity:]
      For all $d_1, d_2 \in D, d_1 \oplus d_2 = d_2 \oplus d_1$.
    \item [Associativity:]
      For all $d_1, d_2, d_3 \in D, d_1 \oplus (d_2 \oplus d_3) = (d_1 \oplus d_2) \oplus d_1$.
    \item [Idempotence:]
      For all $d \in D, d \oplus d = d$.
  \end{description}
  \fi
   \ A partial order, denoted by $\sqsubseteq$, is induced on the elements of $\calD$ as follows:
  for all $d_1, d_2 \in D, d_1 \sqsubseteq d_2$ iff $d_1 \oplus d_2 = d_2$.
  A combine semilattice is \emph{complete} if it is closed under infinite combines.
\end{definition}

\begin{definition}[GFA\cite{SAGA:MW91,LNCS:Ramalingam96}]
	\label{De:GFA}
Let $\calD = (D, \combine)$ be a complete combine semilattice.
Recall that in a regular-tree grammar $G=(N,\Sigma,S,\delta)$,
$\delta$ is a  set of productions of the form
\[
  X_0 \to g(X_1, \ldots, X_k),\text{~~~}\text{~~~~~~~ with }g \in \Sigma.
\]

In a GFA problem $\calG = (G, \calD)$,
each production is associated with a \emph{production function}
$\semgamma\cdot$
that provides an interpretation of $g$---i.e.,
$\semgamma{g} \colon D^k\to D$. \iffull\footnote{
  The definition above is a simplified version of GFA.
  In the usual definition, interpretations are given via productions
  rather than alphabet symbols.
  That approach is somewhat more expressive because
  if  two productions use the same
  symbol, e.g., $X_0 \to g(X_1, X_2)$ and $X_3 \to g(X_4, X_5)$,
  the production functions for the two productions are allowed to
  be different.
  We use the simplified definition because we do not need
  this ability.
}\fi
$\semgamma{\cdot}$ is extended to trees in $L(G)$ in the usual way,
by thinking of each tree $e \in L(G)$ as a term over the
operations $\semgamma{g}$.
Term $e$ denotes a composition of functions, and
corresponds to a unique value in $D$, which we call $\semgamma{e}_\calG$
(or simply $\semgamma{e}$ when $\calG$ is understood).

Let $L_G(X)$ denote the trees derivable from a nonterminal $X$.
The \emph{grammar-flow-analysis problem} is to overapproximate, for each
nonterminal $X$, the \emph{combine-over-all-derivations}
value $m_\calG(X)$ defined as follows:
\[
  m_\calG(X) = \Combine_{e \in L_G(X)} \semgamma{e}_\calG.
\]

We can also associate $G$ with a system of mutually recursive equations,
where each equation has the form
\begin{eqnarray}
	\label{Eq:GFA}
	n_\calG(X_0) = \Combine_{X_0 \to g(X_1, \ldots, X_k) \in \delta} \semgamma{g}(n_\calG(X_1), \ldots, n_\calG(X_k)).
\end{eqnarray}
We use $n_\calG(X)$ to denote the value of nonterminal $X$ in the
\emph{least fixed-point solution} of $G$'s equations.
\end{definition}

In essence, GFA is about two ways of folding the semantics of terms
onto nonterminals:
\begin{description}
  \item [{\it Derivation-tree based:}]
    $m_\calG(X)$ defines the semantics of a term in a compositional fashion, and folds all terms in
    $L_G(X)$ onto nonterminal $X$ by combining ($\oplus$) their values.
  \item [{\it Equational:}]
    $n_\calG(X)$ obtains a value for $X$ by using the values of ``neighboring''
    nonterminals---i.e.,  nonterminals that appear on the right-hand side of productions
    of $X$.
\end{description}
Furthermore, GFA ensures that for all $X$, $m_\calG(X)\sqsubseteq n_\calG(X)$. 

The relevance of GFA for showing unrealizability is that whenever an RTG $G$ is recursive, $L(G)$ is an infinite set of trees. Thus, in general, there is not a clear method to compute the combine-over-all-derivations value $m_\calG(X)=\bigoplus_{e\in L(G)} \semgamma{e}_\calG$. However, we can employ fixed-point finding procedures to compute $n_\calG(X)$. 
Because $m_\calG(X)\sqsubseteq n_\calG(X)$, our computed value will be a safe overapproximation.

However, in some cases we have a stronger relationship between $m_\calG(X)$ and $n_\calG(X)$.
A production function $\semgamma{g}$ is \emph{infinitely distributive} in
a given argument position if
\[
 \semgamma{g}(\ldots, \Combine\limits_{j \in J} x_j, \ldots ) = \Combine\limits_{j \in J} \semgamma{g}(\ldots, x_j, \ldots )
\]
where $J$ is a finite or infinite index set.

\begin{theorem}\cite{SAGA:MW91,LNCS:Ramalingam96}\label{The:Coincidence}
If every production function $\semgamma{g}$, $g\in \Sigma$, is infinitely
distributive in each argument position, then for all nonterminals $X$,
$m_\calG(X) = n_\calG(X)$.\iffull\footnote{\theoref{Coincidence} generalizes other similar theorems
	\cite{ActaInf:KU77,PFA:SP81} about the coincidence of the
	valuations obtained from a path-based semantics (generalized
	in GFA to the derivation-tree-based semantics $\{ m_\calG(X) \mid X \in N \}$)
	and an equational semantics ($\{ n_\calG(X) \mid X \in N \}$) when
	dataflow functions distribute over the combine operator.}\fi
\end{theorem}
This theorem is key to our decision procedures for LIA and CLIA grammars, because the domain of semi-linear sets has this property (\sectref{SemiLinearSets}).

\subsection{Connecting GFA to Unrealizability}
\label{Se:gfaandunrealizability}

In this section, we show how GFA can be used to 
check whether a \sygus problem with finitely many examples $E$ is unrealizable.
Intuitively, we use GFA to overapproximate the set of  values the expressions generated by the grammar
can yield  when evaluated on a certain set of input examples $E$.

\begin{definition}
\label{def:gfa-from-examples}
Let $\textit{sy}^E =(\psi^E, G)$ be a \sygus problem with example set $E$,
regular-tree grammar $G=(N,\Sigma,S,\delta)$, and background theory $T$.
Let $\sem{\cdot}_E$ be the semantics of trees in $L_G(X)$ obtained
via $T$, when $\mu_E(\cdot)$ is used to interpret occurrences of terminals
of $G$ that represent arguments to the function to be synthesized in
the \sygus problem.

Let $\calD = (D,\oplus)$ be a complete combine semilattice for which
there is a concretization function $\gamma \colon D \to
\textit{Val}^{|E|}$, where $\textit{Val}$ is the type of the output
values produced by the function to be synthesized in the \sygus problem.
Let $\calG_E = (G, \calD)$ be a GFA problem that uses
$\mu_E(\cdot)$ to interpret occurrences of terminals of $G$ that represent
arguments to the function to be synthesized.
Then
\begin{enumerate}
  \item
    $\calG_E$ is a \emph{sound abstraction} of the semantics of $L_G(X)$ if
    \[
      \gamma(m_{\calG_E}(X)) \supseteq \{ \sem{e}_E \mid e \in L_G(X) \}.
    \]
  \item
    $\calG_E$ is an \emph{exact abstraction} of the semantics of $L_G(X)$ if
    \[
      \gamma(m_{\calG_E}(X)) = \{ \sem{e}_E \mid e \in L_G(X) \}.
    \]
\end{enumerate}
\end{definition}

By using such abstractions, including the one
described in \sectref{overview} based on semi-linear sets (see
\sectrefs{ProvingUnrealizabilityOfSyGuSProblemsInLIA}{SolvingSyGuSProblemsInCLIA}),
the results obtained by solving a GFA problem can imply that a \sygus
problem with finitely many examples $E$ is unrealizable.

The idea is that, given a \sygus problem $\textit{sy}^E =(\psi^E, G)$ with example set $E$,
regular-tree grammar $G=(N,\Sigma,S,\delta)$, and background theory $T$, we can
(i) solve the GFA problem $\calG_E = (G, \calD)$ with some complete
domain semilattice $\calD = (D,\oplus)$ to obtain an overapproximation of
$\gamma(m_{\calG_E}(S))$, and then
(ii) check if the approximation is disjoint from the specification,
	i.e., the predicate $\vec{o} \in \gamma(m_{\calG_E}(S))\wedge \bigwedge_{i_j \in E} \psi(\vec{o}_j, i_j)$ is unsatisfiable.
	
Checking that the previous predicate holds can be operationalized with
the use of \emph{symbolic concretization} \cite{VMCAI:RSY04} and an
SMT solver. We view an abstract domain $\calD$ as (implicitly) a logic
fragment $\LogicLang_\calD$ of some general-purpose logic
$\LogicLang$,
and each abstract value as (implicitly) representing a formula in $\LogicLang_\calD$.
The connection between $\calD$ and $\LogicLang_\calD$ can be made explicit: we say that $\gammaHat$ is a \emph{symbolic-concretization operation}
for $\calD$ if $\gammaHat(\cdot,\vec{o}) : \calD\rightarrow \LogicLang_\calD$ maps each $a \in \calD$ to a formula with free variables $\vec{o}$, such that $\eval{\gammaHat(a,\vec{o})}_\LogicLang = \gamma(a)$. If $\gammaHat$ exists, we say that \emph{$\LogicLang$ supports
symbolic concretization for $\calD$}.
	\begin{theorem}
		\label{The:GFAisSoundAndComplteForUnrealizable}
		
		Let $\textit{sy}^E =(\psi^E, G)$ be a \sygus problem with example set $E$,
		regular-tree grammar $G=(N,\Sigma,S,\delta)$, and 
		background theory $T$.
		Let $\calD = (D,\oplus)$ be a complete combine semilattice, and 
		$\calG_E = (G, \calD)$ be a grammar-flow-analysis problem over regular-tree grammar $G$.
		Assume the theory $T$ supports symbolic concretization of $\calD$.
		Let $\calP$ be the property
		\[
		\calP \eqdef  \gammaHat(n_{\calG_E}(S), \vec{o}) \wedge \bigwedge_{i_j \in E} \psi(\vec{o}_j, i_j).
		\]
		\begin{enumerate}
			\item\begin{changebar}
				Suppose that $\calG_E$ is a \textbf{sound} abstraction of the semantics of $L(G)$
			\end{changebar}
			with respect to background theory $T$.
			Then $\textit{sy}^E$ is unrealizable \textbf{if}\ \ $\calP$ is unsatisfiable.
			\item
				\begin{changebar}
					Suppose that $\calG_E$ is an \textbf{exact} abstraction of the semantics of $L(G)$
				\end{changebar}
			with respect to background theory $T$.
			Then $\textit{sy}^E$ is unrealizable \textbf{if and only if}\ \ $\calP$ is unsatisfiable.
		\end{enumerate}
%		Furthermore, if\hspace{0.5mm} $T$ supports symbolic concretization of $\calD$, then $sy^E$ is unrealizable if\ \ $\gammaHat(n_{\calG_E}(S), \vec{o}) \land \bigwedge_{i_j\in E} \psi(o_j, i_j)$ is unsatisfiable. Moreover, if $\calD$ is an exact abstraction for $sy^E$ and infinitely distributive, then $sy^E$ is unrealizable iff $\gammaHat(n_{\calG_E}(S), \vec{o}) \land \bigwedge_{i_j\in E} \psi(o_j, i_j)$ is unsatisfiable.
		%GFA is \textbf{sound} for proving the unrealizability of  $sy^E$.
		%Moreover, if $D$ is an exact abstraction domain for $sy^E$, GFA is also \textbf{complete}.
	\end{theorem}
	\iffull
	\begin{proof}
		Suppose 	$\gammaHat(n_{\calG_E}(S), \vec{o}) \land \bigwedge_{i_j\in E} \psi(o_j, i_j)$ is unsatisfiable. By definition of symbolic concretization this means $\not\exists \vec{o} \in \eval{\gammaHat(n_{\calG_E}(S),\vec{o})}_\LogicLang$ such that $\bigwedge_{i_j\in E} \psi(o_j, i_j)$. Equivalently 
		\[
		\Forall{\vec{o} \in \gamma(n_{\calG_E}(S))} \bigvee_{i_j \in E} \neg \psi(o_j, i_j).
		\]
		Since $m_{\calG_E}(X)\sqsubseteq n_{\calG_E}(X)$, the above implies
		\[
		\Forall{\vec{o} \in \gamma(m_{\calG_E}(S))} \bigvee_{i_j \in E} \neg \psi(o_j, i_j).
		\]
		Since $\calG_E$ is a sound abstraction we have 
		\[
		\Forall{\vec{o} \in \{ \sem{e}_E \mid e \in L_G(S) \}} \bigvee_{i_j \in E} \neg \psi(o_j, i_j).
		\]
		This means that for every possible output vector of start symbol $S$ there is one coordinate that violates the specification. Thus, the problem is unrealizable.
		
		Furthermore, if $\calG_E$ is an exact abstraction, and infinitely distributive, the above properties are all equivalent. Thus, the above chain of reasoning also goes in the reverse direction.
	\end{proof}
	\fi

\subsection{Algorithm for Showing Unrealizability}
\label{Se:algorithm}
\algref{unrealizability} summarizes our strategy for showing unrealizability.

\begin{changebar}
\begin{algorithm}[tb]
{\small
  \SetKwInOut{Input}{Input}
  \SetKwInOut{Output}{Output}
  \SetKwInOut{Function}{Function}
  \Function{\textsc{CheckUnrealizable}$(G,\psi,E)$}
  \Input{Grammar $G$, specification $\psi$, set of examples $E$}
  %\Output{Whether problem $\sy^E$  is unrealizable where $\sy= (\psi(f(\vec{x}),\vec{x}), G)$}
  $\calG_E\gets(G, \calD) $ \hspace*{\fill} // GFA problem from $G$ and $E$ (Def.~\ref{def:gfa-from-examples}) \label{Li:FirstLine}\;
  $s \gets n_{\calG_E}(\Start)$ \hspace*{\fill}// Compute solution to the GFA problem\label{Li:solveGFA}\;
  \If{$\gammaHat(s, \vec{o}) \land \bigwedge_{i_j\in E} \psi(o_j, i_j)$ is unsatisfiable\label{Li:smtCheck}}{
  	\Return Unrealizable
  }
  \Return  $\begin{cases}
  	\text{Realizable},& \calG_E \text{ is an exact abstraction} \\
  	\text{Unknown},& \text{otherwise}
  \end{cases}$
 % Generate new example input $i$, add it to $E$, and go to \lineref{FirstLine}\;
}
\caption{Checking whether $\sy^E$ is unrealizable}
\label{Alg:unrealizability}
\end{algorithm}
\end{changebar}

\begin{example}
Recall the \sygus problem, from \sectref{overview}, of synthesizing a function $e_f(x)= 2x + 2$ using the grammar from \eqref{grammar1}. Suppose that we call \algref{unrealizability} with the example set $E=\{1\}$, and use the abstract domain of semi-linear sets. \algref{unrealizability} first creates a GFA problem $\calG_E$, which is shown as the recursive equation system given as \eqref{gfa1subst}. The solution of the GFA problem then gets assigned to $s$ at \lineref{solveGFA}. In this example, $s$ is the semi-linear set $\{0+\lambda3\}$. This set can be 
symbolically concretized as the set of models of $\exists \lambda \ge 0. o_1 = 0 + \lambda 3$. Then, on \lineref{smtCheck} the LIA formula $\exists \lambda \ge 0. o_1 = 0 + \lambda3\land o_1 = 2i_1 + 2 \land i_1 = 1$ is passed to an SMT solver, which will return unsat.
\end{example}

\paragraph{GFA in Practice.}
So far we have been vague about how GFA problems are computationally solved. In general, there is no universal method. 
The performance and precision of a method depends on the choice of abstract domain $\calD$. 

\mypar{Kleene iteration.} Traditionally one would employ Kleene iteration to find a least fixed-point, $n_{\calG_E}(X)$. 
However, Kleene iteration is only guaranteed to converge to a least fixed-point if the domain 
$\calD$ satisfies the finite-ascending-chain condition. 
For example, the domain of predicate abstraction has this property, and therefore \algref{unrealizability} could be instantiated with Kleene iteration and predicate abstraction to attempt to show unrealizabilty, for arbitrary \sygus problems. 
However, in this paper we are focused on \sygus problems using integer arithmetic, which does have infinite ascending chains. 
Thus, while predicate abstraction, and other domains with finite height, can provide a \textbf{sound} abstraction of LIA problems, they can never provide an \textbf{exact} abstraction. 
Alternatively, we could still use Kleene iteration on a domain with infinite ascending chains if we provide a \emph{widening} operator, to ensure convergence \cite{POPL:CH78}. 
The issue with this strategy is that we are not guaranteed to achieve a \emph{least} fixed-point. 
Such a method would still be sound, but necessarily incomplete.

\mypar{Constrained Horn clauses.} Another incomplete, but general, method would employ the use of the domain of constrained Horn clauses, $(\Phi, \lor)$. The set $\Phi$ contains all first-order predicates over some theory. The order of predicates is given by $P_1(\vec{v})\le P_2(\vec{v})$ iff $P_1(\vec{v})\rightarrow P_2(\vec{v})$, for all models $\vec{v}$.
The production functions $\semgamma{\cdot}$ of this GFA problem get translated to constraints on the predicates. 
The advantage of using $(\Phi, \lor)$ is that the resulting GFA problem is a Horn-clause program, which we can then pass to an
off-the-shelf, incomplete Horn-clause solver, such as the one implemented in Z3 \cite{z3}.
In this case, \algref{unrealizability} would be slightly modified. 
Horn-clause solvers do not provide an abstract description of the nonterminals.
Instead they determine satisfiabilty of a set of Horn clauses with respect to a particular query. Therefore, in this case \algref{unrealizability} would use the formula in \lineref{smtCheck} as the Horn-clause query, instead of having a separate SMT check.
\begin{example}
\label{ex:horn}
The GFA problem in \eqref{gfa1} can be encoded using the following constrained Horn clause:
\begin{equation}
  \label{Eq:SimplestProblemAsHorn}
  \forall v,v'.\ \Start(v) \leftarrow (v = 1+1+1+v' \land \Start(v')) \vee v=0
\end{equation}
A Horn-clause solver can prove that the LIA \sygus problem from \sectref{overview}
is unrealizable by showing that the following formula is unsatisfiable:
$\mbox{\eqref{SimplestProblemAsHorn}} \wedge \Start(o_1) \wedge o_1=2i_1+2$.
\end{example}

\mypar{Newton's Method.} In the next two sections, we provide specialized \emph{complete} methods to solve GFA problems over LIA and CLIA grammars using Newton's method~\cite{EsparzaKL10}. Our custom methods are limited to the case of LIA and CLIA grammars, but we show that the resulting solution is exact. No prior method has this property for LIA and CLIA grammars. Consequently, our methods guarantee that not only does the check on \lineref{smtCheck} imply unrealizability on a set of examples if the solver returns unsat, but also realizability if the solver returns sat. The latter property is important because it ensures that the current set of examples is insufficient to prove unrealizability, and we must generate more.

%%% Local Variables:
%%% mode: latex
%%% TeX-master: "main.tex"
%%% End:

% -*- TeX-master: t; TeX-PDF-mode: t -*-

\section{Proving Unrealizability of LIA \sygus Problems with Examples}
\label{Se:ProvingUnrealizabilityOfSyGuSProblemsInLIA}

In this section, we instantiate the framework underlying \algref{unrealizability} to obtain a 
\emph{decision procedure} for (un)realizability of \sygus problems in
\emph{linear integer arithmetic} (LIA),
where the specification is given by examples (as defined in~\exref{LIA}).
First, we review the conditions for applying Newton's method for finding
the least fixed-point of a GFA problem over a commutative, idempotent, $\omega$-continuous semiring
(\sectref{NewtonMethod}). 
We then show that the domain of semi-linear sets can be formulated as such a problem. 
This approach provides a method to compute $n_{\calG_E}(\Start)$ for LIA \sygus problems. 
We then show that the domain of semi-linear sets is \emph{exact} and
\emph{infinitely distributive} (\sectref{SemiLinearSets}).
Finally, we show that semi-linear sets admit symbolic concretization (\sectref{LIAchecksat}). 
Thus, by \theoref{GFAisSoundAndComplteForUnrealizable}, we obtain a
decision procedure for checking (un)realizability.

\subsection{Solving Equations using Newton's Method}
\label{Se:NewtonMethod}

We provide background definitions on semirings
and Newton's method for solving equations over 
certain semirings.

\begin{definition}
\label{De:SemiringSimple}
A \emph{semiring} $\Se = (D, \combine, \extend, \srZero, \srOne)$
consists of a set of \emph{elements} $D$ equipped with two binary operations:
\emph{combine} ($\combine$) and \emph{extend} ($\extend$).
$\combine$ and $\extend$ are associative, and have identity elements
$\srZero$ and $\srOne$, respectively.
$\combine$ is commutative, and $\extend$ distributes over $\combine$.
For every $x \in D$, $x \otimes \srZero = \srZero = \srZero \otimes x$.

A semiring is \emph{commutative} if
for all $a,b \in D$, $a \extend b = b \extend a$.

\iffull
An \emph{$\omega$-continuous semiring} is a semiring with the following additional properties:
\begin{enumerate}
  \item 
    The relation $\sqsubseteq \,\,\, \eqdef \{(a,b) \in D \times D \mid \Exists{d} a \combine d = b \}$
    is a partial order.
  \item\label{def:wcont-supremum}
    Every $\omega$-chain $(a_i)_{i\in \nat}$
    (i.e., for all $i \in \nat~~a_i \sqsubseteq a_{i+1}$)
    has a supremum $\sup_{i\in \nat} a_i$ with respect to $\sqsubseteq$.
  \item\label{def:wcont-inf-distributive}
    Given an arbitrary sequence $(c_i)_{i \in \nat}$, define
    \[
      \wsum c_i \eqdef \sup \{c_0 \combine c_1 \combine \ldots \combine c_i \mid i \in \nat\}.
    \]
    The supremum exists by (\ref{def:wcont-supremum}) above.
    Then, for every sequence $(a_i)_{i \in \nat}$,
    for every $b \in \Se$, and every partition $(I_j)_{j \in J}$ of $\nat$,
    the following properties all hold:
    \[
    \begin{array}{ccccc}
      b \extend \left(\wsum a_i \right) =\wsum \left(b \extend a_i \right) &~\quad~&
      \left(\wsum a_i \right) \extend b = \wsum \left(a_i \extend b\right) &~\quad~&\\
      {\displaystyle \Combine\limits_{j \in J} \left( {\displaystyle \Combine\limits_{i \in I_j} a_i} \right)} = \wsum a_i
    \end{array}
    \]
\end{enumerate}
The notation $a^i$ denotes the $i^{\textit{th}}$ term in the sequence in which
$a^0 = \srOne$ and $a^{i+1} = a^{i} \extend a$.
\fi
An $\omega$-continuous semiring has a \emph{Kleene-star operator}
$\vphantom{a}^{\ostar} \colon D \rightarrow D$ defined as follows:
$a^{\ostar} = \wsum a^i$.

A semiring is \emph{idempotent} if for all $a \in D$, $a \combine a = a$.
\iffull
In an idempotent semiring, the order on elements is defined by
$a \sqsubseteq b$ iff $a \combine b = b$.
\fi
\end{definition}

%\else
%The definition of a semiring can be found in the long version of the paper
%submitted as supplementary material.
%In this paper, the focus is on semirings in which $\combine$ is
%\emph{idempotent} (i.e., for all $a \in D$, $a \combine a = a$).
%In an idempotent semiring, the order on elements is defined by
%$a \sqsubseteq b$ iff $a \combine b = b$.
%Moreover, we are interested in  $\omega$-continuous semirings,
%which are equipped with a \emph{Kleene-star operator}
%$\vphantom{a}^{\ostar} \colon D \rightarrow D$ defined as follows:
%$a^{\ostar} = \wsum a^i$.
%\fi

\iffull
Recently, there has been renewed interest in solving equations over
semirings, with applications to static program analysis.
\emph{Kleene iteration}---the standard iterative approach to solving equations
in program analysis---can be used, but converges to the least fixpoint
only when the semiring has no infinite ascending chains.
\else
Recently,
\fi
\citet{EsparzaKL10} developed an iterative method, called
\emph{Newtonian Program Analysis} (NPA), which solves a set of
semiring equations by an iterative computation.
\iffull
The technique does not  operate on the equations
themselves, but on an augmented set of expressions
created using a notion of a formal derivative of the expressions
on the equation system's right-hand sides.
\fi

\begin{lemma}\label{Lem:NPAForCISemirings}[Newton's Method \cite[Theorem 7.7]{EsparzaKL10}]
  For a system of equations in $N$ variables over a commutative,
  idempotent, $\omega$-continuous semiring,
  NPA reaches the least fixed point after at most $|N|$ iterations.
\end{lemma}

\lemref{NPAForCISemirings} is a powerful result because it applies
even in cases when the semiring has infinite ascending chains.

\subsection{Removing Non-Commutative Operators}
\label{Se:commutativeGrammars}

Our first step towards using GFA to generate equations that can be solved using Newton's method removes
non-commutative operators from the grammar.

We define the language $\LIAPLUS$,
\[
  T_{\LIAPLUS} ::= \textrm{Plus}(T_{\LIAPLUS},T_{\LIAPLUS}) \mid \textrm{Num}(c) \mid \textrm{Var}(x) \mid \textrm{NegVar}(x)
\]
\iffull
with the following semantics with respect to examples $E$:	
\noindent
\begin{tabular}{l@{\hspace{3.0ex}}l}
	\begin{minipage}{2.5in}
		\begin{eqnarray}
		\label{Eq:LIAPLUSPlus}
		\sem{\textrm{Plus}}_E(v_1,v_2)  & := & v_1 + v_2  \\
		\label{Eq:LIAPLUSNum}
		\sem{\textrm{Num}(c)}_E        & := & \B{c,...,c}
		\end{eqnarray}
	\end{minipage}
	\\
	\begin{minipage}{2.5in}
		\begin{eqnarray}
		\label{Eq:LIAPLUSVar}
		\sem{\textrm{Var}(x)}_E    & := & \mu_E(x) \\
		\label{Eq:LIAPLUSNegVar}
		\sem{\textrm{NegVar}(x)}_E & := & -\mu_E(x)
		\end{eqnarray}
	\end{minipage}
\end{tabular}	

\else
where the semantics of the \textrm{Plus}, \textrm{Num}, and \textrm{Var} operators are the same as for
\LIA, and $\sem{\textrm{NegVar}(x)}_E :=  -\mu_E(x)$.
\fi	
We say a regular-tree grammar is an $\LIAPLUS$ grammar if its alphabet is $\{ \textrm{Plus},
\textrm{Num}(c), \textrm{Var}(x), \textrm{NegVar}(x) \}$.

\iffull
We next show how any LIA grammar can be rewritten into
an LIA$^+$  grammar that accepts terms that are semantically equivalent to those in the original grammar.
We introduce a grammar-rewriting function $\mkadditive$ that
recursively pushes negations to the leaves of the terms in an LIA
grammar $G$, to produce an LIA$^+$ grammar $\mkadditive(G)$ that does
not contain the Minus symbol.
Given an LIA grammar $G=(N,\Sigma,S^\LIA,\delta)$, 
we define the rewritten grammar $\mkadditive(G)$ as the tuple $(N\cup N^-,\Sigma^\LIAPLUS,S,\delta^-)$ where 
$\delta^-$ is defined as follows.
For every production $X\to \alpha\in\delta$:
\begin{itemize}
	\item If $\alpha=\textrm{Plus}(X_1,X_2)$, then
	 $\delta^-$ contains the productions $X^-\to\textrm{Plus}(X_1^-,X_2^-)$ and $X\to\textrm{Plus}(X_1,X_2)$;
	\item If $\alpha=\textrm{Minus}(X_1,X_2)$, then
	 $\delta^-$ contains the productions $X^-\to\textrm{Plus}(X_1^-,X_2)$ and $X\to\textrm{Plus}(X_1,X_2^-)$;
	\item If $\alpha=\textrm{Num}(c)$, then
	 $\delta^-$ contains the productions $X\to \textrm{Num}(c)$ and $X^-\to \textrm{Num}(-c)$.
	\item If $\alpha=\textrm{Var}(x)$, then
	 $\delta^-$ contains the productions $X\to \textrm{Var}(x)$ and $X^-\to \textrm{NegVar}(x)$.	
\end{itemize} 
\else
	The next example shows how
	our algorithm uses a function $\mkadditive$ to push negations to the leaves of LIA terms to yield an LIA$^+$ grammar.
\fi
\iffull
It is trivial to see that the grammar $\mkadditive(G)$ only produces terms in LIA$^+$.
\fi
\begin{example}
Consider the LIA grammar $G$:
\[
\begin{array}{rcl}
\Start	&::= &\textrm{Minus}(\Start,\Start)\mid 1 \mid x\\
\end{array}
\]
The following LIA$^+$ grammar $h(G)$ is equivalent to $G$:
\[
\begin{array}{rcl}
\Start	&::= &\textrm{Plus}(\Start,\Start^-)\mid \textrm{Num}(1) \mid \textrm{Var}(x)\\
\Start^-	&::= &\textrm{Plus}(\Start^-,\Start)\mid\textrm{Num}(-1) \mid \textrm{NegVar}(x).
\end{array}
\]
\end{example}

\iffull
The following lemma shows that the original and the rewritten grammars produce
semantically equivalent terms.
\begin{lemma}
\label{Lem:rewritesound}
	An LIA grammar $G$ is semantically equivalent to the LIA$^+$ grammar $\mkadditive(G)$, i.e.,
	\begin{eqnarray}
		&\left(\forall e\in L(G)\exists e'\in L(\mkadditive(G)).\sem{e}=\sem{e'}\right)\\
		&\wedge\left(\forall e'\in L(\mkadditive(G))\exists e\in L(G).\sem{e}=\sem{e'}\right).
	\end{eqnarray}
%$$\forall e.e\in L(G)\leftrightarrow \exists e'\in L(h(G)).\sem{e}=\sem{e'}$$
\end{lemma}
\begin{proof}
	
	We start by proving the following result, which states that the terms produced by some nonterminal $X$ in $G$ are equivalent to terms produced by the corresponding nonterminal $X$ in $\mkadditive(G)$, and 
	to the negation of terms produced by the corresponding negative nonterminal $X^-$ in $\mkadditive(G)$:
	\rone $\left(\forall e\in L_G(X)\exists e'\in L_{\mkadditive(G)}(X).\sem{e}=\sem{e'}\right)\wedge\left(\forall e'\in L_{\mkadditive(G)}(X)\exists e\in L_G(X).\sem{e}=\sem{e'}\right),$ and \rtwo $\left(\forall e\in L_G(X)\exists e'\in L_{\mkadditive(G)}(X^-).\sem{e}=\sem{e'}\right)\wedge\left(\forall e'\in L_{\mkadditive(G)}(X^-)\exists e\in L_G(X).\sem{e}=-\sem{e'}\right).$
	
	We proceed by induction on $e$.
	The base case are $e=\textrm{Num}(c)$ and $e=\textrm{Var}(x)$. According to the definition of $h$, there exists productions $X\to\textrm{Num}(c)$ (resp., $X\to\textrm{Var}(x)$) and $X\to\textrm{Num}(-c)$ (resp., $X\to\textrm{NegVar}(x)$) in $h(G)$. Note that $\sem{\textrm{Num}(c)}=-\sem{\textrm{Num}(-c)}$ and $\sem{\textrm{Var}(x)}=-\sem{\textrm{NegVar}(x)}$.
	Hence, the base case holds.
	
	Now the induction step is
	\begin{itemize}
		\item Assume $e=\textrm{Plus}(e_1,e_2)$ where $e_1$ and $e_2$ are terms produced by nonterminals $X_1$ and $X_2$, respectively. 
		According to the induction hypothesis, $X_1$ in $h(G)$ can produce a term $e_1'$ equivalent to $e_1$ and $X_2$ in $h(G)$ can produce a term $e_2'$ equivalent $e_2$. Therefore the nonterminal $X$ in $h(G)$ can produce $\textrm{Plus}(e_1',e_2')$ whose semantic is equivalent to $e$. The analysis for $X^-$ in $h(G)$ is similar.
		\item Assume $e=\textrm{Minus}(e_1,e_2)$ where $e_1$ and $e_2$ are terms produced by nonterminals $X_1$ and $X_2$, respectively. 
		According to the induction hypothesis, $X_1$ in $h(G)$ can produce a term $e_1'$ equivalent to $e_1$ and $X^-_2$ in $h(G)$ can produce a term $e_2'$ such that $\sem{e_2'}=-\sem{e_2}$ Therefore the nonterminal $X$ in $h(G)$ can produce $\textrm{Plus}(e_1',e_2')$ whose semantic is equivalent to $e$, i.e., $\sem{\textrm{Plus}(e_1',e_2')}=\sem{e_1'}-\sem{e_2}=\sem{\textrm{Minus}(e_1,e_2)}$. The analysis for $X^-$ in $h(G)$ is similar.
	\end{itemize}
At last, terms produced by $\Start$ in $G$ are semantically equivalent to terms produced by $\Start$ in $h(G)$, and hence $G$ is semantically equivalent to $\mkadditive(G)$
\end{proof}
 \fi

\subsection{Grammar Flow Analysis Using Semi-Linear Sets}
\label{Se:SemiLinearSets}

Thanks to \sectref{commutativeGrammars}, we can assume that the
\sygus grammar $G$ only produces LIA$^+$ terms.
In this section, we use grammar-flow analysis to generate equations
such that  the solutions to the  equations assign a semi-linear set to each
nonterminal $X$ that, for the finitely many examples in $E$,
\emph{exactly} describes the set of possible values produced by
any term in $L_G(X)$.

We start by defining the complete combine semilattice
$(\slint, \oplus)$ of \emph{semi-linear sets}
(see \cite[\S2.3.3]{EsparzaKL10} and \cite[\S3.4.4]{POPL:BET03}).
We then use them, together with the set of examples $E$,
to define a specific family of GFA problems:
$\calG_E = (G, \slint)$, where $G = (N,\Sigma,S,\delta)$ is an
$\LIAPLUS$ grammar. For simplicity, we use notation $\slint$ for both the semilattice and its domain
%For every nonterminal $X$ in grammar $G$,
%the solution to the GFA problem is a semi-linear set $m_{\calG_E}(X)$
%that \emph{exactly} describes the set of possible output vectors
%produced by evaluating each term in $L_G(X)$ 
%on
%the examples in $E$.

In the terminology of abstract interpretation, $\slint$ is
an abstract domain that we can use to represent, for every nonterminal $X$,
the set of possible output vectors produced by evaluating each term in $L_G(X)$ on
the examples in $E$.
Moreover, the representation is \emph{exact}; i.e.,
$\gamma(m_{\calG_E}(X)) = \{ \sem{e}_E \mid e \in L_G(X) \}$
where $\gamma$ denotes the usual operation of concretization.

\begin{definition}[Semi-linear Set]
  A \textit{linear  set} $\linset{\vec{u},\{\vec{v}_1, \cdots, \vec{v}_n\}}$ denotes
  the set of integer vectors $\{\vec{u}+\lambda_1\vec{v}_1 + \cdots + \lambda_n\vec{v}_n \mid \lambda_1,\ldots,\lambda_n \in \nat\}$,
  where $\vec{u},\vec{v}_1,...,\vec{v}_n\in\mathbb{Z}^d$ and $d$ is the dimension of the linear set.
  A \textit{semi-linear set} is a finite union $\bigcup_i \linset{\vec{u}_i,V_i}$ of linear sets,
 also denoted by $\{ \langle \vec{u}_i, V_i \rangle \}_i$.

  The \emph{concretization} of a semi-linear set $\textit{sl} = \{ \langle \vec{u}_i, V_i \rangle \}_i$,
  denoted by $\gamma(\textit{sl})$, is the set of vectors
  \[
    \bigcup_i \{\vec{u}_i + \lambda_{1,i}\vec{v}_{1,i} + \cdots + \lambda_{n,i}\vec{v}_{n,i} \mid \lambda_{1,i}, \ldots, \lambda_{n,i} \in \nat\}.
  \]
\end{definition}

\iffull
Semi-linear sets were originally used in a well-known result in
formal-language theory: Parikh's theorem~\cite{Parikh66}.
Parikh's theorem states that, given a context-free grammar $G$ with
terminals $(t_1,\ldots, t_n)$, if one looks only at the number of
occurrences of each terminal symbol in each word in a context-free
language, without regard to their order---i.e., each word $w$ is
represented by a vector $v_w = \B{c_1,\ldots, c_n}$, which denotes that each terminal $t_i$
appears exactly $c_i$ times in $w$---the set of vectors $\{v_w \mid w\in \lang(G)\}$
is representable by a semi-linear set.
If a grammar for an LIA \sygus problem only uses addition (which is a commutative operation),
we can represent any term in the language of the grammar by simply counting the number of times
each terminal (i.e., a constant or a variable) appears in the term.
Consequently, we can use a domain of values similar to the ones used
in Parikh's theorem to represent the set of possible terms (or, more precisely, their
semantics) as a semi-linear set.

While the details of Parikh's theorem are not relevant to this paper,
the core idea behind its proof is that grammars over commutative operators can be transformed
into regular languages and therefore regular expressions.
Then, to compute the set of all possible count vectors that the grammar can produce
one needs to ``evaluate'' the regular expressions using operators analogous to
the regular-expression concatenation, union, and star.
For semi-linear sets, these operators are $\otimes$, $\oplus$ and $\ostar$,
defined as follows \cite[\S3.4.4]{POPL:BET03}:
\else
To apply Newton's method for solving equations (\lemref{NPAForCISemirings}),
we need a commutative idempotent semiring over semi-linear sets.
Fortunately, such a semiring exists \cite[\S3.4.4]{POPL:BET03}, with
the operators $\otimes$, $\oplus$ and $\ostar$ defined as follows:
\fi
\begin{align}
  \hspace{-1mm}\{\langle\vec{u}_{1,i},V_{1,i}\rangle\}_i\oplus\{\langle\vec{u}_{2,j},V_{2,j}\rangle\}_j  & =  \{\langle\vec{u}_{1,i},V_{1,i}\rangle\}_i\cup\{\langle\vec{u}_{2,j},V_{2,j}\rangle\}_j   \notag \\
  \label{Eq:SLExtend}
  \hspace{-1mm}\{\langle\vec{u}_{1,i},V_{1,i}\rangle\}_i\otimes\{\langle\vec{u}_{2,j},V_{2,j}\rangle\}_j  & =  \bigcup_{i,j}\{\langle \vec{u}_{1,i}+\vec{u}_{2,j},V_{1,i}\cup V_{2,j}\rangle\}\ \ \ \ \ \ \ \ \   \notag
\\
  (\{\langle\vec{u}_{i},V_{i}\rangle\}_i)^{\ostar} & =  \{ \langle \vec{0},\bigcup_{i}(\{\vec{u}_i\}\cup V_i)\rangle \}  \raisetag{3mm}
\end{align}
The semi-linear sets 
$\mathbf{0} \eqdef \emptyset$ and $\mathbf{1} \eqdef \{\langle \vec{0},\emptyset\rangle\}$ are the identity
elements for $\oplus$ and $\otimes$, respectively.
We use $(\slint, \oplus)$ to denote the complete combine semilattice of semi-linear sets\iffull with the least element $\mathbf{0}$\fi.

We define the GFA problem
$\calG_E = (G, \slint)$ by giving the following
interpretations to $\LIAPLUS$ operators:
\begin{eqnarray}
  \label{Eq:LIAPLUSGFAPlus}
  \semgamma{\textrm{Plus}}_E(sl_1,sl_2) \hspace{-3mm} & = &\hspace{-3mm} sl_1\otimes sl_2 \\
  \label{Eq:LIAPLUSGFANum}
  \semgamma{\textrm{Num}(c)}_E         \hspace{-3mm} & = &\hspace{-3mm} \{\langle\B{c,\cdots,c},\emptyset\rangle\} \\
  \label{Eq:LIAPLUSGFAVar}
  \semgamma{\textrm{Var}(x)}_E         \hspace{-3mm} & = &\hspace{-3mm} \{\langle \mu_E(x),\emptyset\rangle\} \\
  \label{Eq:LIAPLUSGFANegVar}
  \semgamma{\textrm{NegVar}(x)}_E      \hspace{-3mm} & = &\hspace{-3mm} \{\langle -\mu_E(x),\emptyset\rangle\}
\end{eqnarray}
Now consider the combine-over-all-derivations value $m_{\calG_E}(X)=\bigoplus_{e \in L_G(X)} \semgamma{e}_E$
for the grammar-flow-analysis problem $\calG_E$.
For an arbitrary tree $e \in L_{G}(X)$, in the computation of $\semgamma{e}_E$ via
\eqseqref{LIAPLUSGFAPlus}{LIAPLUSGFANegVar}, there is never any use of the $\oplus$ operation of $\slint$.
Consequently, the computation of $\semgamma{e}_E$ produces a semi-linear set that consists of
a \emph{single vector}---the same vector, in fact, that is produced by the computation of
$\sem{e}_E$ \iffull via \eqseqref{LIAPLUSPlus}{LIAPLUSNegVar}.
\else shown in \exref{LIA}.
\fi 
\begin{changebar}
In particular, $\oplus$ two lines above \eqref{SLExtend} preserves singleton sets, and hence for singleton sets, $\otimes$ one line above
\eqref{SLExtend} emulates 
\end{changebar}\iffull\eqref{LIAPLUSPlus}.
\else
$\sem{\textrm{Plus}}_E$.
\fi
Therefore, the combine-over-all-derivations value $m_{\calG_E}(X) = \bigoplus_{e \in L_G(X)} \semgamma{e}_E$
is exactly the set of vectors $\{ \sem{e}_E \mid e \in L_G(X) \}$.
In other words, $m_{\calG_E}(X)$ is an \emph{exact} abstraction of the $\sem{\cdot}_E$ semantics of the
terms in $L_{G}(X)$, i.e., $\gamma(m_{\calG_E}(X)) = \{ \sem{e}_E \mid e \in L_G(X) \}$.
Because  $\semgamma{\textrm{Plus}}_E$ is infinitely distributive over $\oplus$
(\cite[Defn.\ 2.1 and \S2.3.3]{EsparzaKL10}),
$m_{\calG_E}(X) = n_{\calG_E}(X)$ holds by \theoref{Coincidence},
and thus we can compute $m_{\calG_E}(X)$ by solving a set of equations
in which, for each $X_0 \in N$, there is an equation of the form
\begin{eqnarray}
  \label{Eq:eqs}
  n_{\calG_E}(X_0)\hspace{-1mm} = \Combine_{X_0 \to g(X_1, \ldots, X_k) \in \delta}\hspace{-1mm} \semgamma{g}_E(n_{\calG_E}(X_1), \ldots, n_{\calG_E}(X_k)).
\end{eqnarray}

\iffull
The argument given in the previous paragraph is captured by the following lemma:

	\begin{lemma}
		\label{Le:slisexact}
		 Given an LIA$^+$ grammar $G = (N,\Sigma,S,\delta)$, a finite set of examples $E$,
		 $\calG_E=(G, \Se)$ is an exact abstraction of the semantics of the languages $L_G(X)$, for all $X \in N$ (with respect
		 to LIA and $E$).
	\end{lemma}
	\begin{proof}We can show that for any expression
		$e$, the abstract semantics  $\semgamma{e}_E$ is always a singleton set $\{\sem{e}_E\}$,
		where the element of the singleton set is exactly the semantics of $e$. 
		For an arbitrary tree $e \in L_{G}(X)$, in the computation of $\semgamma{e}_E$ via
		\eqseqref{LIAPLUSGFAPlus}{LIAPLUSGFANegVar}, there is never any use of the $\oplus$ operation of $\slint$.
		Consequently, the computation of $\semgamma{e}_E$ produces a semi-linear set that consists of
		a \emph{single vector}---the same vector, in fact, that is produced by the computation of
		$\sem{e}_E$ \iffull via \eqseqref{LIAPLUSPlus}{LIAPLUSNegVar}.
		\else shown in \exref{LIA}.
		\fi 
		In particular, \eqref{SLExtend} preserves singleton sets, and hence for singleton sets,
		\eqref{SLExtend} emulates \iffull\eqref{LIAPLUSPlus}.
		\else
			$\sem{\textrm{Plus}}_E$.
		\fi
		Therefore, the combine-over-all-derivations value $m_{\calG_E}(X) = \bigoplus_{e \in L_G(X)} \semgamma{e}_E$
		is exactly the set of vectors $\{ \sem{e}_E \mid e \in L_G(X) \}$.
		In other words, $m_{\calG_E}(X)$ is an \emph{exact} abstraction of the $\sem{\cdot}_E$ semantics of the
		terms in $L_{G}(X)$, i.e., $\gamma(m_{\calG_E}(X)) = \{ \sem{e}_E \mid e \in L_G(X) \}$.

		Therefore, $\calG_E$ is an exact abstraction of the semantics of $L_G(X)$.
	\end{proof}
\fi

\begin{example}
Consider again the $\LIAPLUS$ grammar $G_1$ from \eqref{grammar1}, written out in
the expanded form given in \footnoteref{ExpandedGrammar}\iffull:
\[
\begin{array}{c}
  \Start	 ::=  \textrm{Plus}(\SOne, \Start) \mid \textrm{Num}(0) 
  \hspace{5mm}
  \SOne ::= \textrm{Plus}(\STwo, \textrm{Var}(x))\hspace{5mm}\\
   \STwo ::= \textrm{Plus}(\SThree, \textrm{Var}(x))\hspace{5mm}
  \SThree ::= \textrm{Var}(x).
\end{array}
\]
\else
.
\fi
Let $E$ be $\{ 1, 2 \}$, and thus $\mu_E(x) = \B{1,2}$.
The equation system for the GFA problem $\calGOne_E$ is as follows:
\[
  \begin{array}{@{\hspace{0ex}}r@{\hspace{1.0ex}}c@{\hspace{1.0ex}}l@{\hspace{0ex}}}
    &&\hspace{-4mm}n_{\calGOne_E}(\Start)   =  n_{\calGOne_E}(\SOne) \otimes n_{\calGOne_E}(\Start) \oplus \{\langle (0,0), \emptyset \rangle\}    \\
        &&\hspace{-4mm}n_{\calGOne_E}(\SOne)    =  n_{\calGOne_E}(\STwo) \otimes \{\langle (1,2), \emptyset \rangle\}  \\
        &&\hspace{-4mm}n_{\calGOne_E}(\STwo)    =  n_{\calGOne_E}(\SThree) \otimes \{\langle (1,2), \emptyset \rangle\} \enspace \hspace{1.5mm}
        n_{\calGOne_E}(\SThree)  =  \{\langle (1,2), \emptyset \rangle\}
  \end{array}
\]
which has the solution
\[
  \begin{array}{@{\hspace{0ex}}r@{\hspace{1.0ex}}c@{\hspace{1.0ex}}l@{\hspace{3.0ex}}r@{\hspace{1.0ex}}c@{\hspace{1.0ex}}l@{\hspace{0ex}}}
    n_{\calGOne_E}(\Start)  & = & \{\langle (0, 0), \{ (3,6) \} \rangle\}  &    n_{\calGOne_E}(\STwo)   & = & \{\langle (2,4), \emptyset \rangle\} \\
    n_{\calGOne_E}(\SOne)   & = & \{\langle (3,6), \emptyset \rangle\}  &     n_{\calGOne_E}(\SThree) & = & \{\langle (1,2), \emptyset \rangle\}.
  \end{array}
\]
The concretizations of semi-linear sets in the solution are
\[
  \begin{array}{@{\hspace{0ex}}r@{\hspace{1.0ex}}c@{\hspace{1.0ex}}l@{\hspace{0ex}}}
    \gamma(n_{\calGOne_E}(\Start))  & = & \{(0, 0) + \lambda(3,6) \mid \lambda \in \nat \} \}  \\
    \gamma(n_{\calGOne_E}(\SOne))   & = & \{ (3,6) \}  \qquad \gamma(n_{\calGOne_E}(\STwo))    =  \{ (2,4) \} \\	
    \gamma(n_{\calGOne_E}(\SThree))  & = &  \{ (1,2) \}.
  \end{array}
\]
\end{example}

%\subsection{Solving Equations for LIA$^+$ Grammars}
%\label{Se:equationsinliaplus}

The following proposition shows that the equations generated in \eqref{eqs} can be solved using 
Newton's method.

\begin{proposition}
  $(\slint,\oplus,\otimes,\mathbf{0},\mathbf{1})$ is a commutative,
  idempotent, $\omega$-continuous semiring. 
\end{proposition} 

\iffull
Moreover, $(\slint,\oplus,\otimes,\mathbf{0},\mathbf{1})$ has infinite ascending chains;
consequently, \lemref{NPAForCISemirings} is directly relevant to our setting.
Henceforth, we use the term ``semiring''---and symbol $\Se$---to mean a
\emph{commutative, idempotent, $\omega$-continuous semiring}.
\fi

\iffull
Before concluding this section, we analyze the size of the semi-linear set computed by the NPA method
when solving equations generated by LIA$^+$ grammars.
\fi
For a semi-linear set $sl=\{\linset{\vec{u}_i,V_i}_i\}$, let its \emph{size} be $\sum_i(|V_i|+1)$. 
%\iffull
%\begin{theorem}
%\label{The:sizesl}
%Given an LIA (or LIA$^+$)  grammar $G=(N,\Sigma,S,\delta)$, a finite set of examples $E$ and a nonterminal $X\in N$, 
%NPA yields a semi-linear set $n_{\calG_E}(X)$  of size $2^{2^{O(n\log (n))}}$ where $n:=|N|$.
%\end{theorem}
%\heping{proof missing}
%\else
%For a given LIA grammar $G=(N,\Sigma,S,\delta)$ and a finite set of examples $E$, the semi-linear set $n_{\calG_E}(S)$ computed by NPA has size $2^{2^{O(n\log (n))}}$, where $n:=|N|$.
%\fi
Given an LIA   grammar , a finite set of examples $E$ and a nonterminal $X\in N$, 
the semi-linear set $n_{\calG_E}(X)$ yielded by NPA can contain exponentially many linear sets \cite{kopczynski2010parikh}. 
%In \sectref{evaluation}, we evaluate NPA on real \sygus benchmarks and demonstrate that how the size of \sygus problems effect the size of computed semi-linear sets.

\subsection{Checking Unrealizability}
\label{Se:LIAchecksat}

We now show how symbolic concretization for $\slint$ can be used
to prove that no element $\vec{o}$ in $n_\calG(\Start)$ 
satisfies the specification $\psi^E(\vec{o})$ of the \sygus problem.
The logic $\LIA$  supports symbolic concretization for $\slint$.
For instance, for a linear set $\{\langle\vec{u},\{\vec{v}_1,\ldots,\vec{v}_n\}\rangle\}$,
its symbolic concretization $\gammaHat( \langle\vec{u},\{\vec{v}_1,\ldots,\vec{v}_n\}\rangle, \vec{o} )$ is defined as follows:
\[
  \exists \lambda_1 \in \nat, \ldots, \lambda_n \in \nat . (\vec{o} = \vec{u} + \lambda_1\vec{v}_1 + \cdots + \lambda_n\vec{v}_n).
\]
Thus, the symbolic concretization for a semi-linear set  is:
\begin{eqnarray}
	\label{Eq:gammahatsl}
	\gammaHat( \{\langle\vec{u}_i,V_i\rangle\}_i, \vec{o}) \eqdef \bigvee_i \gammaHat(\langle\vec{u}_i,V_i\rangle, \vec{o}).
\end{eqnarray}
Note that $\vec{o}$ is shared among all disjuncts.
\iffull
The set of satisfying assignments to $\vec{o}$ consist of exactly the vectors in
$\gamma( \{\langle\vec{u}_i,V_i\rangle\}_i)$.
\fi

Our decidability result follows directly from \theoref{GFAisSoundAndComplteForUnrealizable}.
\begin{theorem}
  Given an LIA \sygus problem $\sy$ and a finite set of examples $E$,
  it is decidable whether the \sygus problem $\sy^E$ is realizable.
\end{theorem}
\iffull
\begin{proof}
	We have shown that $n_{\calG_E}(x)$  is an exact abstraction for LIA grammars (Lemma~\ref{Le:slisexact}) and LIA supports symbolic concretization (\eqref{gammahatsl}). According to \theoref{GFAisSoundAndComplteForUnrealizable}, $\calG_E = (G, \Se)$ is sound and complete for proving unrealizability of LIA \sygus problems for finitely example, and hence decidable.
\end{proof}
\fi

%%% Local Variables: 
%%% mode: latex
%%% TeX-master: "main.tex"
%%% End: 

% -*- TeX-master: t; TeX-PDF-mode: t -*-

\section{Proving Unrealizability of CLIA \sygus Problems with Examples}
\label{Se:SolvingSyGuSProblemsInCLIA}

In this section, we instantiate the framework from \sectref{gfa}  to obtain a 
\emph{decision procedure} for realizability of \sygus problems in
\emph{conditional linear integer arithmetic} (CLIA), where the specification is given by examples.
The decision procedure follows the same steps as the one for LIA in \sectref{ProvingUnrealizabilityOfSyGuSProblemsInLIA}.
The main difference is a technique for solving equations generated from
grammars that involve both Boolean and integer operations.

\subsection{Conditional Linear Integer Arithmetic}
\label{Se:CLIAdef}
%CLIA terms add conditionals to LIA terms.
The grammar of all CLIA terms is the following:
\[ 
\begin{array}{rcl}
T_\integ & ::= & \textrm{IfThenElse}(T_{\bool}, T_{\integ},T_{\integ}) \mid \textrm{Plus}(T_{\integ},T_{\integ}) \\
&\mid&  \textrm{Minus}(T_{\integ},T_{\integ}) \mid Num(c) \mid Var(x)\\
T_\bool & ::= & \textrm{And}(T_{\bool}, T_{\bool}) \mid \textrm{Not}(T_\bool)  \mid \textrm{LessThan}(T_\integ,T_\integ) 
\end{array}
\]
where $c\in \integ$ is a constant and $x\in \var$ is a input variable to the function being synthesized.
Notice that the definitions of $T_\integ$ and $T_\bool$ are mutually recursive. \iffull\footnote{For any \sygus problem over CLIA terms, the goal will
be to synthesize a term with a specific type---i.e., either $\bool$ or $\integ$.}\fi
The example grammar presented in \eqref{grammar2} in \sectref{overview} is a CLIA grammar.

We now define the semantics of CLIA terms.
Given an integer vector $\vec{v}\in\integ^d$ and a  Boolean vector $\vec{b}\in\bool^d$,
let $\projectint(\vec{v},b)$ be the integer vector obtained by
keeping the vector elements of $\vec{v}$ corresponding to the indices
for which $\vec{b}$ is true, and zeroing out all other elements:
\begin{eqnarray*}
  \lefteqn{\projectint(\B{u_1,\ldots,u_d}, \B{b_1,\ldots,b_d})}\\
    & = &\hspace{-2mm} \B{\text{if}(b_1)\text{ then }u_1\text{ else }0,\ldots,\text{if}(b_d)\text{ then }u_d\text{ else }0}
\end{eqnarray*}

\noindent The semantics of symbols that are not in LIA is as follows:
\[
  \begin{array}{@{\hspace{0ex}}r@{\hspace{1.0ex}}c@{\hspace{1.0ex}}l@{\hspace{0ex}}}
    \sem{\textrm{IfThenElse}}_E(\vec{b},\vec{v_1},\vec{v_2}) & = & \projectint(\vec{v_1},\vec{b})+\projectint(\vec{v_2},\neg \vec{b})\\
    \sem{\textrm{Not}}_E(\vec{b}) = \neg \vec{b}  &&   \sem{\textrm{And}}_E(\vec{b_1},\vec{b_2}) = \vec{b_1}\wedge \vec{b_2} \\
    \sem{\textrm{LessThan}}_E(\vec{v_1},\vec{v_2})           & = & \vec{v_1}< \vec{v_2}
  \end{array}
\]
where the operations $+$, $\wedge$, $<$, and $\neg$ are performed element-wise---e.g.,
$\vec{u}< \vec{v}=\B{b_1,\ldots,b_n}$ such that $b_i\Leftrightarrow u_i<v_i$.

Similarly to what we did in \sectref{commutativeGrammars}, any CLIA grammar $G$
can be rewritten into an equivalent $\CLIAPLUS$ grammar $\mkadditive(G)$ that does not contain any
occurrences of Minus, but may contain the symbol NegVar.
%as the tuple $(N\cup N^-,\Sigma^\CLIAPLUS,S,\delta^-)$ where 
%\begin{eqnarray*}
%	\Sigma^\CLIAPLUS=\{\textrm{Plus}, \textrm{Var}(x), \textrm{NegVar}(x), \textrm{Num}(c), \textrm{IfThenElse}, \textrm{And}\\
%	, \textrm{Not}, \textrm{LessThan}\},
%\end{eqnarray*}
%and $\delta^-$ contains all the productions described in \sectref{commutativeGrammars} 
%as well as the productions
%introduced by the following rule:
%\begin{itemize}
%\item For each production $X\to\textrm{IfThenElse}(X_1,X_2,X_3)$ in $\delta$, then $\delta^-$ contains the productions $X\to\textrm{IfThenElse}(X_1,X_2,X_3)$ and $X^-\to \textrm{IfThenElse}(X_1,X_2^-,X_3^-)$.
%\end{itemize}
%We can show that this transformation is sound using a proof identical to that from \lemref{rewritesound}.
%We thus assume that grammars do not contain the symbol Minus, but may contain the symbol  NegVar.

The rest of the section is organized as follows. 
First, we present the abstract domains used to represent Boolean and integer terms (\sectref{abstractsemanticsCLIA}).
Second, we
show how to compute an exact abstraction of Boolean nonterminals in grammars without IfThenElse (\sectref{booleanonly}).
Third,  we show how to solve \sygus problems with CLIA grammars containing arbitrary operators, in particular IfThenElse and mutual recursion (\sectref{mutualrecursion}). 

\subsection{Abstract Semantics for CLIA}
\label{Se:abstractsemanticsCLIA}

We use sets of Boolean vectors as the abstract domain for Boolean nonterminals,
and semi-linear sets as the abstract domain for integer nonterminals. 
We use  $b$ to denote a Boolean vector and $\bset$ to denote sets of Boolean
vectors.

Given a semi-linear set $sl{\in}\slint$ and a Boolean vector $\vec{b}{\in}\bool^d$,
let $\project(sl,\vec{b})$ be the semi-linear set obtained by
zeroing out \iffull for each vector in $sl$ the \fi elements at all index positions for which $\vec{b}$ is false:
\[
\begin{array}{@{\hspace{0ex}}r@{\hspace{1.0ex}}c@{\hspace{1.0ex}}l@{\hspace{0ex}}}
    \project(\{\langle\vec{u}_{i},\Omega_{i}\rangle\}_i, \vec{b})      & = & \{\projectls(\langle\vec{u}_{i},\Omega_{i}\rangle, \vec{b})\}_i\\
    \projectls(\linset{\vec{u},\{\vec{v}_1,...,\vec{v}_n\}}, \vec{b}) & = & \linset{\projectint(\vec{u},\vec{b}),\{\projectint(\vec{v}_i,\vec{b})\}_i}
\end{array}
\]

Next, we lift the concrete semantics to semi-linear sets
and define the abstract semantics of CLIA operators\iffull that are not in LIA\fi.
\[
\begin{array}{@{\hspace{0ex}}r@{\hspace{1.0ex}}c@{\hspace{1.0ex}}l@{\hspace{0ex}}}
  \semgamma{\textrm{IfThenElse}}_E(\bset,sl_1,sl_2) & = & \\
		\multicolumn{3}{r}{\bigoplus_{\vec{b}\in \bset}
				\project(sl_1,\vec{b})\otimes \project(sl_2,\neg \vec{b})} \\[1mm]
  \semgamma{\textrm{LessThan}}_E(sl_1,sl_2)        & = & \{v_1{<}v_2 \mid v_1\in sl_1,v_2\in sl_2\} \\
  \semgamma{\textrm{Not}}_E(\bset)                  & = & \bigcup_{\vec{b}\in \bset}\{\neg \vec{b}\} \\
  \semgamma{\textrm{And}}_E(\bset_1,\bset_2)         & = & \bigcup_{\vec{b_1}\in \bset_1,\vec{b_2}\in \bset_2}\{\vec{b_1}\wedge \vec{b_2}\}
\end{array}
\]

\iffull
\begin{example}
	Consider a set of Boolean vectors $\bset:=\{(\tr,\fa),(\tr,\tr)\}$ and two semi-linear sets $sl_1:=\{\linset{(1,2),\{(3,4)\}}\}$ and $sl_2:=\{\linset{(5,6),\{(7,8)\}}\}$. Then \rone
	$\semgamma{\textrm{Not}}_E(\bset)=\{(\fa,\tr),(\fa,\fa)\}$, and \linebreak
	\rtwo $\semgamma{\textrm{LessThan}}_E(sl_1,sl_2)=\{(\tr,\tr),(\tr,\fa),(\fa,\fa)\}$ since $(1,2)<(5,6)=(\tr,\tr)$, $(1,2)+(3,4)<(5,6)=(\tr,\fa)$ and $(1,2)+2(3,4)<(5,6)=(\fa,\fa)$.
	Finally,
\begin{eqnarray*}
  &&\semgamma{\textrm{IfThenElse}}_E(\bset,sl_1,sl_2)\\
  & = & \{\linset{(1,0),\{(3,0)\}}\}\otimes\{\linset{(0,6),\{(0,8)\}}\} \\
&&\oplus\{\linset{(1,2),\{(3,4)\}}\}\otimes \{\linset{(0,0),\{(0,0)\}}\}\\
\hspace{-2mm} & = &\hspace{-2mm}\{\linset{(1,6),\{(3,0),(0,8)\}},\linset{(1,2),\{(3,4),(0,0)\}}\}
\end{eqnarray*}\qed
\end{example}
\fi

\noindent
Operationally, the semantics of the LessThan symbol can be implemented
using an SMT solver.
As shown in \sectref{LIAchecksat}, a semi-linear set $sl$ can be
symbolically concretized as a formula $\gammaHat(sl,\vec{o})$ in LIA
(a decidable SMT theory).
Therefore, the set $\semgamma{\textrm{LessThan}}_E(sl_1,sl_2) = \bset$ can
be computed by performing $2^{|E|}$ SMT queries---i.e.,
for every Boolean vector $\vec{b} = \B{b_1,\ldots,b_{|E|}}$, we have that
$\vec{b}\in \bset$ iff the following formula is satisfiable:
$\gammaHat(sl_1,\vec{o}_1)\wedge \gammaHat(sl_2,\vec{o}_2)\wedge \vec{b}=\vec{o_1}<\vec{o_2}$.

Similarly to how we defined $\semgamma{\cdot}_E$ for multisorted
terms, we overload $\oplus$ as the union of sets of Boolean vectors,
and define a multisorted semilattice 
$\calD_{\text{CLIA}^+}:=(2^{\mathbb{B}}\uplus\slint,\oplus)$ over sets of
Boolean vectors and semi-linear sets.
We use $\calGclia := (G,\calD_{\text{CLIA}^+})$ to denote the GFA problem
for a CLIA$^+$ grammar $G$ and finitely many examples $E$.
$\calGclia$ is an exact abstraction of the semantics of CLIA$^+$ grammars.
\iffull
\begin{lemma}
	\label{Le:cliaexact}
  Given CLIA$^+$ grammar\ \ $G = (N,\Sigma,S,\delta)$, finite set of examples $E$,
  $\calGclia$ is an exact abstraction of the semantics of the languages $L_G(X)$, for all $X \in N$ (with respect
  to LIA and $E$).
\end{lemma}
\begin{proof}
  Using a similar argument as in \sectref{SemiLinearSets}, we can show that for any expression
  $e$, the abstract semantics  $\semgamma{e}_E$ is always a singleton set $\{\sem{e}_E\}$,
  where the element of the singleton set is exactly the semantics of $e$.
  Therefore, $m_{\calGclia}=\bigoplus_{e\in L_G(X)}\semgamma{e}_E$ is exactly $\{\sem{e}_E\mid e\in L_G(X)\}$. 
\end{proof}
\fi

\subsection{CLIA Equations Without Mutual Recursion}
\label{Se:booleanonly}
A CLIA grammar $G$ contains Boolean and integer nonterminals.
A nonterminal $X$ is a Boolean nonterminal if $\sem{X}\in\bool$,
and is an integer nonterminal if $\sem{X}\in \integ$.
In this subsection, we assume that there exists no mutual recursion,
i.e., $G$ contains no IfThenElse productions.
Under this assumption, the only operator that connects Boolean nonterminals
and integer nonterminals is LessThan, and hence no Boolean nonterminal
appears in the productions of an integer nonterminal.
Therefore, we can proceed by first solving the equations that involve integer nonterminals,
using the technique presented in \sectref{NewtonMethod},
and then plugging the corresponding values into the equations that involve
Boolean nonterminals.
\iffull
\begin{example}
\label{Exa:nomutual}
Consider the following grammar $G_b$:
\begin{align}
\hspace{-2mm}
\begin{array}{r@{~}r@{~~}ll}	
BExp		&::=	& \textrm{LessThan}(X, N2) \mid \textrm{LessThan}(N0, Exp)\\
&\mid&
\textrm{And}(BExp,BExp)\\
Exp		&::=	&	 \textrm{Plus}(X, Exp) \mid \textrm{Num}(0)\qquad
X			::=   \textrm{Var}(x) \\
N0			&::=&	   \textrm{Num}(0)\qquad
N2			::=	   \textrm{Num}(2)\\
\end{array}
\label{Eq:grammar3}
\end{align}
Assume that the given set of examples is $E=\{1,2\}$.
If we consider the equations generated by grammar flow analysis for this grammar, 
all the variables corresponding to the integer nonterminals $Exp, X,N0,N2$ do not
depend on any of the variables for the Boolean nonterminals.
Therefore, we can solve the corresponding
set of equations using the techniques presented in \sectref{ProvingUnrealizabilityOfSyGuSProblemsInLIA}. For each such nonterminal $X$,
by plugging the value of each $n_\calGclia(X)$ in the equations corresponding to $BExpr$ we get the 
following equation:
\begin{align}
\hspace{-3mm}
\begin{array}{r@{~}r@{~~}ll}
n_\calGclia(BExp)		&=&	 \{(\tr, \fa)\} \oplus \{(\tr, \tr),(\fa, \fa)\}\\
&&\oplus
\semgamma{\textrm{And}}(n_{\calGclia}(BExp),n_{\calGclia}(BExp))
\end{array}
\label{Eq:booleanex}
\end{align}
where $\oplus$ is the set union operator.
\qed
\end{example}
\fi
After this step, we are left with a set of equations $\eqs_{\bool}$
that involve only Boolean nonterminals and Boolean symbols.
\iffull
Concretely, for every nonterminal $X$ in the 
set of Boolean nonterminals $N_{\bool}$, $\eqs_{\bool}$ contains an equation 
\begin{eqnarray}
n_{\calGclia}(X)=\smashoperator{\bigoplus_{X\to g(X_1,...,X_k)\in \delta }}~~\semgamma{g}_E (n_{\calGclia}(X_1),...,n_{\calGclia}(X_k))
\end{eqnarray} 
\fi
Because the domain of sets of Boolean vectors is finite, 
the least fixed point of $\eqs_\bool$ can be found using an  algorithm \iffull\textsc{SolveBool}\fi
that iteratively computes finer under-approximations of $n_{\calGclia}$ as $n_{\calGclia}^{\texttt{k}}$---i.e., the under-approximation
at iteration $\texttt{k}$---until it reaches the least fixed point, which---by \theoref{Coincidence}---is an exact abstraction.
\iffull
The initial under-approximation is $n_{\calGclia}^{(0)}(X)=\emptyset$ for all Boolean nonterminals in $X$.
The  under-approximation
of each terminal $X$ at iteration $k$ is the following expression:
\begin{eqnarray*}
	&&n_{\calGclia}^{(i)}(X)\\
	&=&n_{\calGclia}^{(i-1)}(X)\oplus\\
	&& \bigoplus_{X\to g(X_1,...,X_n)\in \delta }\semgamma{g}_{E} (n_{\calGclia}^{(i-1)}(X_1),...,n_{\calGclia}^{(i-1)}(X_n)) \mid X\in N_{\bool}).
\end{eqnarray*}
Notice that $\semgamma{g}_{E}$ is computable for every operator $g$ (\sectref{abstractsemanticsCLIA}).
\fi
This algorithm terminates in at most $2^{|E|}|N_\bool|$ iterations
because the set of Boolean vectors has size at most $2^{|E|}$,
and each iteration adds at least one Boolean vector to 
one of the variables until the least fixed point is reached. 
\iffull
\begin{example}
Recall \eqref{booleanex} from \exref{nomutual}.
In the first iteration of the iterative algorithm $n_{\calGclia}^{(0)}(X)=\emptyset$.
We compute $n_{\calGclia}^{(1)}(BExp)$ as follows:	
\begin{eqnarray*}
	n_{\calGclia}^{(1)}(BExp)	&	=&	 \{(\tr, \fa)\} \oplus \{(\tr, \tr),(\fa, \fa)\}
	\\&\oplus&\semgamma{\textrm{And}}(n_{\calGclia}^{(0)}(BExp),n_{\calGclia}^{(0)}(BExp))
\end{eqnarray*}

and obtain $n_{\calGclia}^{(1)}(BExp)=\{(\tr, \fa),(\tr, \tr),(\fa, \fa)\}$.
If we compute $n_{\calGclia}^{(2)}(BExp)$ using the same technique, we reach a fixed point---i.e., $n_{\calGclia}^{(2)}(BExp)=n_{\calGclia}^{(1)}(BExp)$.
\qed
\end{example}

\begin{lemma}
\label{Le:boolean}
Given a set of equations involving only Boolean-nonterminal variables
and representing the abstract semantics of $k$ examples,
 the iterative algorithm
\textsc{SolveBool} computes a fixed-point solution in at most $n2^k$ iterations, where $n$ is the number
of nonterminal variables.
\end{lemma}
\begin{proof}
Note that $n_{\calGclia}^{(i-1)}(X)\subseteq n_{\calGclia}^{(i)}(X)$ for all $i$ and $X$, and the size of a set of Boolean vector with dimension $k$ is at most $2^k$. Then the size of underapproximations is strictly increasing (otherwise the least fixed point is reached) and bounded by $n2^k$, i.e., $\sum_{X\in N}|n_{\calGclia}^{(i-1)}(X)|< \sum_{X\in N}|n_{\calGclia}^{(i)}(X)|\le n2^k$, for all $i$. Therefore, the iteration number $i$ can be at most $n2^k$.	
\end{proof}
\fi

\subsection{CLIA Equations With Mutual Recursion}
\label{Se:mutualrecursion}

We have seen how to 
compute exact abstractions for grammars without mutual recursion, for both integer (\sectref{SemiLinearSets}) 
and Boolean (\sectref{booleanonly}) nonterminals. 
In this section, we show how to handle grammars that involve $\textrm{IfThenElse}$ symbols, 
which introduce mutual recursion between Boolean and integer nonterminals.
See \eqref{mutualrecursion} in \sectref{overview} for an example of
equations that involve mutual recursion.
To solve mutually recursive equations, 
we cannot simply compute the abstraction for one type and use the corresponding values
to compute the abstraction for the other type, like we did in \sectref{booleanonly}. 
However, we show that if we repeat such substitutions in an iterative fashion, 
we obtain an algorithm \textsc{SolveMutual} that computes an exact abstraction for
a grammar with mutual recursion.

At the \texttt{k}-th iteration, for every nonterminal $X$, the algorithm computes an under-approximation $n_{\calGclia}^{\texttt{k}}(X)$
of  $n_{\calGclia}(X)$.
Initially, $n_{\calGclia}^{\texttt{-1}}(X)=\mathbf{0}$ for all nonterminals $X$ of type $\integ$. 
At iteration $\texttt{k}\geq 0$ the algorithm does the following:

\paragraph{Step 1} 
Replace each integer nonterminal $Z$ with the value $n_\calGclia^{\texttt{k-1}}(Z)$ from iteration \texttt{k{-}1}
and use the technique in \sectref{booleanonly} to
compute $n_\calGclia^{\texttt{k}}(B)$ for each Boolean nonterminal
$B$.
\iffull
Formally, for each Boolean nonterminal $B\in N_\bool$ we have the  equation:
 \begin{eqnarray}
n_{\calGclia}^{\texttt{k}}(B)=\smashoperator{\bigoplus_{B\to g(X_1,...,X_n)\in \delta }}~~\semgamma{g}_E (n_{\calGclia}^{(i_1)}(X_1),...,n_{\calGclia}^{(i_n)}(X_n)) 
\end{eqnarray} 
where 
each $i_j$ is equal to $k$ if $X_j\in N_{\bool}$ and $k-1$  if $X_j\in N_{\integ}$.
\fi

\paragraph{Step 2} 
Replace each Boolean nonterminal $B$ with the value $n_\calGclia^{\texttt{k}}(B)$ from \emph{\textbf{Step 1}}
and
compute $n_\calGclia^{\texttt{k}}(Z)$ for each integer nonterminal
$Z$ (see \eqref{iteonlyexpression} in \sectref{overview} for an example).

\iffull
Formally, for each integer nonterminal $Z\in N_\integ$ we have the  equation:
 \begin{eqnarray}
 \label{Eq:integerEqsForRemoveIf}
n_{\calGclia}^{\texttt{k}}(Z)=\smashoperator{\bigoplus_{Z\to g(X_1,...,X_n)\in \delta }}~~\semgamma{g}_E (n_{\calGclia}^{\texttt{k}}(X_1),...,n_{\calGclia}^{\texttt{k}}(X_n)) 
\end{eqnarray} 
where for each  $X_j\in N_{\bool}$,  $n_{\calGclia}^{\texttt{k}}(X_j)$ is the value computed in \emph{\textbf{Step 1}}.
\fi
The equations obtained at \emph{\textbf{Step 2}} only contain integer nonterminals,
but they may contain IfThenElse symbols for which the abstract semantics 
\iffull contains
the $\project$ operator that 
\fi
is not directly supported by the equation-solving technique
presented in \sectref{NewtonMethod}.
In the rest of this section, we present a way to transform the given set of equations 
into a new set of equations that faithfully
describes the abstract semantics of IfThenElse symbols, using only $\otimes$ and
$\oplus$ operations over semi-linear sets.
\iffull
The resulting equations can be solved using the technique presented in
\sectref{NewtonMethod}.
\fi

The iterative algorithm \textsc{SolveMutual} is guaranteed to terminate in $|N|2^{|E|}$ iterations.

\iffull
\begin{lemma}
\label{Le:booleanTimeComplexity}
Given a set of equations involving both Boolean-
and integer-nonterminal variables that
represent the abstract semantics of\ \ $k$ examples,
 the iterative algorithm
\textsc{SolveMutual} computes a fixed-point solution in at most $n2^k$ iterations, where $n$ is the number
of nonterminal variables.
\end{lemma}
\begin{proof}
	In each iteration, at least one of the set $n_\calGclia^{\texttt{k}}(B)$ of Boolean vectors should be different from the set $n_\calGclia^{\texttt{k-1}}(B)$ in the previous iteration. Each set of Boolean vectors can be only updated at most $2^{|E|}$ times. Therefore there can be at most $|N|2^{|E|}$ iterations.
\end{proof}
\fi

\subsubsection*{$\semgamma{\textrm{IfThenElse}}_E$ using Semi-Linear-Set Operations}
\label{Se:itesolving}
\begin{changebar}
	In this section, we show how to solve  equations that involve IfThenElse symbols. 
Recall the definition of the abstract semantics of IfThenElse symbols:
\begin{align*}
\hspace*{-1.5mm}\semgamma{\textrm{IfThenElse}}_E(\bset,sl_1,sl_2)=\bigoplus_{b\in \bset}
				& \project(sl_1,b) \\
				&\otimes\project(sl_2,\neg b)
\end{align*}
%The semantics of IfThenElse symbols 
%involves the $\project$ operator, which cannot be directly described using only semi-linear-set operations.
%Hence, we cannot directly apply Newton's method as presented in~\sectref{ProvingUnrealizabilityOfSyGuSProblemsInLIA}.

In the rest of this section, we show how equations that involve the semantics of
IfThenElse symbols can be rewritten into equations that involve only
$\oplus$ and $\otimes$ operations, so that they can be solved using
Newton's method.
\end{changebar}
For every possible Boolean vector $b$, the new set of
equations contains a new variable $n_\calGclia^{\texttt{k}}(X^b)$, so
that the solution to the set of equations for this variable is
$\project(n_\calGclia^{\texttt{k}}(X), b)$.
 
Let $\eqs$ be a set of equations over a set of integer nonterminals $N$.
We write $x/y$ to denote the substitution of every occurrence of $x$ with $y$.
We generate a set of equations \iffull$\removeif(\eqs)=\eqs'$\else $\eqs'$ \fi over the set of variables $N^{\mathbb{B}^d}$
as follows. 
For every equation $n^{\texttt{k}}_\calGclia(X) = \bigoplus_i \alpha_i$ in $\eqs$ and $b\in\bool^d$, there exists an equation $n^{\texttt{k}}_\calGclia(X^b)=\bigoplus_i\pi_b(\alpha_i)$ in $\eqs'$, where $\pi_b$ applies the following substitution in this order:
\begin{enumerate}
	\item 	For every $X\in N$ and $b'\in \bool^d$, $\pi_b$ applies the substitution $\project(n^{\texttt{k}}_\calGclia(X),b')/ n^{\texttt{k}}_\calGclia(X^{b\wedge b'})$.
	\item 
	For every $X\in N$, $\pi_b$ applies $n^{\texttt{k}}_\calGclia(X)/ n^{\texttt{k}}_\calGclia(X^b)$.
	%if $\alpha$ is of the form $\semgamma{\textrm{Plus}}_E(n_\calGclia(X_1),n_\calGclia(X_2))$
	%then $\eqs'$ contains an equation \linebreak
	%$n_\calGclia(X^b) = \semgamma{\textrm{Plus}}_E(n_\calGclia(X_1^b),n_\calGclia(X_2^b))$ for 
	%each $b\in \mathbb{B}^d$.
	
	\item 
	For any semi-linear set $sl$ appearing in $\eqs$, $\pi_b$ applies the substitution $sl/\project(sl,b)$. Because $sl$ is a constant, this substitution  yields  a constant semi-linear set.
	\end{enumerate}
\begin{changebar}
	\begin{example}
	Figure~\ref{fig:rewriteite} illustrates how \eqref{iteonlyexpression} is rewritten into \eqsref{itesplitting}. We omit equations for variables $n_{2,E}^{\texttt{1}}(\Start^{\{\fa,\fa\}})$ and $n_{2,E}^{\texttt{1}}(\Start^{\{\tr,\fa\}})$ because they do not contribute to the solving of $n_{2,E}^{\texttt{1}}(\Start^{\{\tr,\tr\}})$. After expanding the definition of $\semgamma{\text{IfThenElse}}$, we apply the substitutions to obtain \eqsref{itesplitting}. Substitution 2 is not applied because there are no variables of the form $n_{2,E}^{\texttt{1}}(X)$ after applying substitution 1.
\begin{figure}[h]
	
	\begin{gather*}
	{\small
		\begin{array}{@{\hspace{0ex}}r@{~}r@{~~}l@{\hspace{0ex}}}
		{n_{2,E}^{\texttt{1}}(\Start)} & = & \semgamma{\textrm{IfThenElse}}_E(\{(\tr, \fa)\}, \{(0,0)+\lambda (3,6)\}, \\
		&   & n_{2,E}^{\texttt{1}}(\Start))  \oplus  \{(0,0)+\lambda (2,4)\}  \oplus  \{(0,0)+\lambda (3,6)\}	
		\end{array}
	}\\
	\hspace*{2.95cm}\wideDownarrow \text{\small Generate equations for $Start^b$}  \\
	{\small
		\begin{array}{@{\hspace{0ex}}r@{~}r@{~~}l@{\hspace{0ex}}}
		{n_{2,E}^{\texttt{1}}(\Start^{(\tr,\tr)})}      & = & \pi_{\{\tr,\tr\}}\big(\semgamma{\textrm{IfThenElse}}_E(\{(\tr, \fa)\}, \{(0,0)+\lambda (3,6)\}, \\
		&   & n_{2,E}^{\texttt{1}}(\Start))\big)  \oplus  \pi_{\{\tr,\tr\}}\big(\{(0,0)+\lambda (2,4)\}\big)  \\
		&\oplus&  \pi_{\{\tr,\tr\}}\big(\{(0,0)+\lambda (3,6)\}\big)\\			
		{n_{2,E}^{\texttt{1}}(\Start^{(\fa,\tr)})}      & = & \pi_{\{\fa,\tr\}}\big(\semgamma{\textrm{IfThenElse}}_E(\{(\tr, \fa)\}, \{(0,0)+\lambda (3,6)\}, \\
		&   & n_{2,E}^{\texttt{1}}(\Start))\big)  \oplus  \pi_{\{\fa,\tr\}}\big(\{(0,0)+\lambda (2,4)\}\big)  \\
		&\oplus&  \pi_{\{\fa,\tr\}}\big(\{(0,0)+\lambda (3,6)\}\big)
		\label{Eq:expandITE}
		\end{array}
	}\\
	\hspace*{3.55cm}\wideDownarrow  \text{\small Expand definition of $\semgamma{\text{IfThenElse}}$}\\
	{\small
		\begin{array}{@{\hspace{0ex}}r@{~}r@{~~}l@{\hspace{0ex}}}
		{n_{2,E}^{\texttt{1}}(\Start^{(\tr,\tr)})}      & = & \pi_{\{\tr,\tr\}}\big(\project(\{(0,0)+\lambda (3,6)\},\{\fa,\tr\})\big)  \\
		&\otimes&  \pi_{\{\tr,\tr\}}\big(\project(n_{2,E}^{\texttt{1}}(\Start),(\fa,\tr))\big) \\
		&\oplus&  \pi_{\{\tr,\tr\}}\big(\{(0,0)+\lambda (2,4)\} \big)
		\oplus \pi_{\{\tr,\tr\}}\big(\{(0,0)+\lambda (3,6)\}\big)\\			
		{n_{2,E}^{\texttt{1}}(\Start^{(\fa,\tr)})}      & = &\pi_{\{\fa,\tr\}}\big(\project(\{(0,0)+\lambda (3,6)\},\{\fa,\tr\})\big)  \\
		&\otimes&  \pi_{\{\fa,\tr\}}\big(\project(n_{2,E}^{\texttt{1}}(\Start),(\fa,\tr))\big) \\
		&\oplus&  \pi_{\{\fa,\tr\}}\big(\{(0,0)+\lambda (2,4)\}\big)
		\oplus \pi_{\{\fa,\tr\}}\big(\{(0,0)+\lambda (3,6)\}\big)
		\end{array}
	}\\
	\hspace*{2.58cm}\bigwideDownarrow    \begin{aligned}
	&\text{\small Apply $\project$ to constants}\\
	&\text{\small Apply substitution 1}
	\end{aligned}\\
	{\small
		\begin{array}{@{\hspace{0ex}}r@{~}r@{~~}l@{\hspace{0ex}}}
		{n_{2,E}^{\texttt{1}}(\Start^{(\tr,\tr)})}      & = & \pi_{\{\tr,\tr\}}\big(\{(0,0)+\lambda (3,0)\}\big)  \otimes  n_{2,E}^{\texttt{1}}(\Start^{(\tr,\tr)\wedge(\fa,\tr) })\\
		&\oplus&  \pi_{\{\tr,\tr\}}\big(\{(0,0)+\lambda (2,4)\}\big)
		\oplus \pi_{\{\tr,\tr\}}\big(\{(0,0)+\lambda (3,6)\}\big)\\			
		{n_{2,E}^{\texttt{1}}(\Start^{(\fa,\tr)})}      & = &\pi_{\{\fa,\tr\}}\big(\{(0,0)+\lambda (3,0)\}\big)  \otimes  n_{2,E}^{\texttt{1}}(\Start^{(\fa,\tr)\wedge(\fa,\tr)}) \\
		&\oplus&  \pi_{\{\fa,\tr\}}\big(\{(0,0)+\lambda (2,4)\}\big)
		\oplus \pi_{\{\fa,\tr\}}\big(\{(0,0)+\lambda (3,6)\}\big)
		\end{array}
	}\\
	\hspace*{1.68cm}\wideDownarrow  \text{\small Apply substitution 3}\\
	{\small
		\begin{array}{@{\hspace{0ex}}r@{~}r@{~~}l@{\hspace{0ex}}}
		{n_{2,E}^{\texttt{1}}(\Start^{(\tr,\tr)})}      & = & \{(0,0)+\lambda (3,0)\}  \otimes  n_{2,E}^{\texttt{1}}(\Start^{(\fa,\tr) })\\
		&\oplus&  \{(0,0)+\lambda (2,4)\} 
		\oplus \{(0,0)+\lambda (3,6)\}\\			
		{n_{2,E}^{\texttt{1}}(\Start^{(\fa,\tr)})}      & = &\{(0,0)+\lambda (0,0)\}  \otimes  n_{2,E}^{\texttt{1}}(\Start^{(\fa,\tr)}) \\
		&\oplus&  \{(0,0)+\lambda (0,4)\}
		\oplus \{(0,0)+\lambda (0,6)\}
		\end{array}
	}
	\end{gather*}
	\caption{Rewriting \eqref{iteonlyexpression} into \eqsref{itesplitting}.}
	\label{fig:rewriteite}
\end{figure}
\end{example}
\end{changebar}

\iffull
\begin{lemma}
Given a set of equations $\eqs$ involving only variables $V:=\{n_\calGclia(X)\}_{X\in N}$,
the set of equations $\removeif(\eqs)$ has at most $|V|2^{|E|}$ variables,
and an assignment $\sigma'$ is a solution of $\removeif(\eqs)$ iff
there exists a solution $\sigma$ of $\eqs$ such that $\sigma(n_\calGclia(X))=\sigma'(n_\calGclia(X^{\vec{\tr}}))$ for all $X\in N$.
\end{lemma}
\begin{proof}
The variables in \removeif(\eqs) are of form $n_\calGclia(X^b)$ for $X\in V$ and $b\in 2^{|E|}$. Therefore there are at most $|N|2^{|E|}$ variables in $\removeif(\eqs)$. 
	
To show that every solution to $\eqs$ is also a solution to $V^{\vec{\tr}}:=\{n_\calGclia(X^{\vec{\tr}})\}$ in  $\eqs'$, it is sufficient to prove that for all $b\in\bool^d$, if an assignment $\sigma:V\to\slint$ is a solution to $\eqs$, then the assignment $\sigma'(n_\calGclia(X^b)):=\project(\sigma(n_\calGclia(X)),b)$ is a solution to $\eqs'$.

Actually, equations in $\eqs$ are of the form $X=\bigoplus_i\alpha_i$ (\eqref{integerEqsForRemoveIf}). Therefore, all we need  to show is that $\project(\alpha_i[\sigma],b)=\pi_b(\alpha_i)[\sigma']$ for all $\alpha_i$ and $b\in \bool^d$, where $[\sigma]:=[\text{for each }x\in V.x/\sigma(x)]$. Note that $\alpha$ must be one of the following form:
\begin{itemize}	
	\item if $\alpha=\project(n_\calGclia(X_1),b_1)\otimes \project(n_\calGclia(X_2),b_2)$, we have 
	\begin{eqnarray*}
		\project(\alpha[\sigma],b)\hspace{-2mm} & = &\hspace{-2mm}\project(\sigma(n_\calGclia(X_1)),b_1\wedge b)\\&\otimes& \project(\sigma(n_\calGclia(X_2)),b_2\wedge b)\\
		\hspace{-2mm} & = &\hspace{-2mm}\sigma'(n_\calGclia(X_1^{b_1\wedge b}))\otimes \sigma'(n_\calGclia(X_2^{b_2\wedge b}))\\
		\hspace{-2mm} & = &\hspace{-2mm}\pi_b(\alpha_i)[\sigma']
	\end{eqnarray*}
	\item if $\alpha=n_\calGclia(X_1)\otimes n_\calGclia(X_2)$, we have 
	\begin{eqnarray*}
		\pi_b(\alpha_i)[\sigma']\hspace{-2mm} & = &\hspace{-2mm}\left(n_\calGclia(X_1^b)\otimes n_\calGclia(X_2^b)\right)[\sigma']\\
				      \hspace{-2mm} & = &\hspace{-2mm}\project(\sigma(n_\calGclia(X_1)),b)\\&\otimes& \project(\sigma(n_\calGclia(X_2)),b)\\
				      \hspace{-2mm} & = &\hspace{-2mm}\project(\alpha[\sigma],b),
	\end{eqnarray*}
	since $\project$ is distributive over $\otimes$.
	\item if $\alpha=sl$, it is obvious that $\project(\alpha[\sigma],b)\hspace{-1mm}=\hspace{-1mm}\pi_b(\alpha_i)[\sigma']$.
	\item if $\alpha=n_\calGclia(X_1)$, we have 
	\begin{eqnarray*}
		\project(\alpha[\sigma],b)&=&\project(\sigma(n_\calGclia(X_1)),b)\\
		&=&\sigma'(n_\calGclia(X_1^b))=\pi_b(\alpha_i)[\sigma'].
	\end{eqnarray*}

\end{itemize}
	Therefore any solution to $\eqs$ is a solution to $V^{\vec{\tr}}$ in $\eqs'$.
	
	For the other direction, we need to show that assume $\sigma'$ is a solution to $\eqs'$, $\sigma(\cdot):=\sigma'(\cdot,\vec{tr})$ is a solution to $\eqs$. The argument for this case is similar to the previous one.
\end{proof}
\fi

%\iffull
%Before concluding this section, we analyze the size of the semi-linear set computed by the NPA method
%when solving equations generated by CLIA grammars.
%The technique $\removeif$ generates 
%equations with $|N|2^{|E|}$ nonterminals. 
%Using the analysis 
%from \theoref{sizesl} we get the following theorem.
%\begin{theorem}
%\label{The:cliasizesl}
%Given an CLIA grammar~~$~G=(N,\Sigma,S,\delta)$, a finite set of examples $E$ and a nonterminal $X\in N$, 
%NPA yields a semi-linear set $n_{\calGclia}(X)$  of size $2^{2^{O(n2^{|E|}\log n)}}$ where $n:=|N|$.
%\end{theorem}
%\heping{proof}
%\fi

\subsection{Checking Unrealizability}
\label{Se:CLIAchecksat}

Using the symbolic-concretization technique described in \sectref{LIAchecksat} ,
and the complexities described throughout this section, we obtain the following decidability theorem.
\begin{theorem}
  Given a CLIA \sygus problem $\sy$ and a finite set of examples $E$,
  it is decidable whether the \sygus problem $\sy^E$ is (un)realizable.
\end{theorem}
\iffull
\begin{proof}
	We have shown that $n_{\calGclia_E}$ is an exact abstraction for CLIA grammars (Lemma~\ref{Le:cliaexact}) and the domain of semi-linear sets supports symbolic concretization. Besides, we have shown a sound and complete algorithm \textsc{SolveMutual}  to solve $n_{\calGclia_E}$. According to the \theoref{GFAisSoundAndComplteForUnrealizable}, GFA is sound and complete for proving unrealizability of CLIA \sygus problems for finitely example, and hence decidable.
\end{proof}
\fi

%%% Local Variables: 
%%% mode: latex
%%% TeX-master: "main.tex"
%%% End: 

%\input{otherDomain}
% -*- TeX-master: t; TeX-PDF-mode: t -*-

\section{Implementation}
\label{Se:implementation}

We implemented a tool \name that can return two-sided answers to unrealizability
problems of the form $\sy = (\psi, G)$. 
When it returns \emph{\textbf{unrealizable}}, no 
term in $L(G)$ satisfies $\psi$;
when it returns \emph{\textbf{realizable}}, some $e \in L(G)$ satisfies $\psi$;
\name can also time out. 
\name consists of three components: 
1) a verifier (the SMT solver CVC4~\cite{cvc4}), which verifies 
the correctness of candidate solutions and produces counterexamples, 
\begin{changebar}
2) a synthesizer (ESolver---the enumerative solver introduced in \cite{alur2016sygus}), which synthesizes solutions from examples, and 
\end{changebar}
3) an unrealizability verifier, which proves whether the problem is unrealizable on the current set of examples. 

\begin{changebar}
\algref{CEGIS} shows \name's  \cegis loop. Given a \sygus problem $\sy= (\psi, G)$, \name first initialize $E$ with a random input example 
with values in the range $[-50,50]$(\lineref{initializaeE}),  and then, in parallel,
\ballns{1} calls ESolver to find a solution of $\sy^E$ (\lineref{callESolver}), and
\ballns{2} uses grammar flow analysis (\algref{unrealizability}) to decide whether $\sy^{E\cup E_r}$ is unrealizable (\lineref{checkUnrealizable}), where  $E_r$ is a set of randomly generated temporary examples. Randomly generated examples are used when the problem is
proven to be realizable by GFA, but we do not have a candidate solution $e^*$---ESolver did not return yet---that can be used to issue an SMT query to possibly obtain a counterexample.
During each CEGIS iteration, the following three events can happen:
1) If GFA returns unrealizable, $\name$ terminates and outputs unrealizable (\lineref{returnUnrealizable}). 
2) If GFA returns realizable, \name adds a temporary random example to $E_r$ (\lineref{randomExample}), and reruns
GFA with $E\cup E_{r}$. 
3) If ESolver returns a candidate solution $e^*$, the problem $\sy^{E}$ is realizable.
(ESolver never uses the temporary random examples.)
Therefore, \name kills 
the GFA process and then issues an SMT query to check if $e^*$ is a solution to the \sygus problem $\sy$ (\lineref{checkSolution}): if not, \name
adds a counterexample  to $E$ (\lineref{addCex}) and triggers the next CEGIS iteration, otherwise, \name return $e^*$ as a solution to the given \sygus problem \sy\  (\lineref{returnSolution}).
\end{changebar}

\name currently has two modes: \namehorn and \namesl.

\namehorn implements the constrained-Horn-clauses technique for solving equations 
presented in \sectref{algorithm}, and uses \z3's Horn-clause solver, Spacer~\cite{z3},
to solve the Horn clauses.	
	
\namesl implements the decision procedures presented in
\sectref{ProvingUnrealizabilityOfSyGuSProblemsInLIA} and~\sectref{SolvingSyGuSProblemsInCLIA} for solving LIA and CLIA 
problems.  
\namesl also implements two optimizations:
(i) \namesl eagerly removes a linear set from a semi-linear set whenever it is trivially subsumed by
another linear set; and
(ii) \namesl uses the optimization presented in the following paragraph.

\begin{changebar}
\begin{algorithm}{
	\small
	\DontPrintSemicolon 
	\SetKwInOut{Function}{Function}
	\SetKwBlock{DoParallel}{do in parallel}{end}
	\SetKw{Continue}{continue}
	\SetKw{Kill}{kill}
	\Function{\textsc{Nay$(G,\psi)$}}
	\KwIn{Grammar $G$, specification $\psi$}
	$i \gets \textsc{Random}(-50,50)$\ \ \ \ \ Set of examples $E\gets \{i\}$\label{Li:initializaeE}\;
	\While{\textnormal{True}}
	{
		\DoParallel{
			%\If{\textnormal{ESolver}($G,\psi,E$)}{then-block}
			\ballns{1} \{$e^*\gets$\textsc{ESolver}($G,\psi,E$)\label{Li:callESolver}\;
				\ \ \ \ \ \ \Kill \ballns{2}\;
				\ \ \ \ \ \ \If{$\exists i_{cex}.\neg\psi(\sem{e^*},i_{cex})$\label{Li:checkSolution}}{\ \ \ \ \ \ $E\gets E\cup\{i_{cex}\}$\label{Li:addCex}\;
					\ \ \ \ \ \ \Continue}
				\ \ \ \ \ \ \Else{\ \ \ \  \ \ \Return $e^*$      \label{Li:returnSolution}      \}}
			\ballns{2} \{\label{Li:checkUnrealizable}
			$E_r\gets \emptyset$\;
			\ \ \ \ \ \ \ \While{\textnormal{True}}{
				\ \ \ \ $result\gets$\textsc{CheckUnrealizable}($G,\psi,E\cup E_r$)\;
				\ \ \ \ \If{$result=$\textnormal{Unrealizable}}{
					\ \ \ \  \Kill \ballns{1}\;\ \ \ \ \Return Unrealizable\label{Li:returnUnrealizable}}
				\ \ \ \ $i\gets\textsc{Random}(-50,50)$\;
				\ \ \ \ $E_r\gets E_r\cup\{i\}$\label{Li:randomExample}\;
				\ \ \ \ \Continue				\}
			}
	}}}
	\caption{CEGIS with random examples }
	\label{Alg:CEGIS}
\end{algorithm}
\end{changebar}

\paragraph{Solving GFA Equations via Stratification.}
%The equations that arise in a GFA problem are amenable to the
%standard optimization technique of identifying ``strata'' of dependences among%nonterminals, and solving the equations by finding values for
%nonterminals of lower ``strata'' first, working up to higher strata in
%an order that respects dependences among the equations.
%This idea can be formalized in terms of the strongly connected
%components (SCCs) of a dependence graph, where the nodes
%are the nonterminals of $G$ and
%the edges represent the dependence of a left-hand-side %nonterminal on a right-hand-side
%nonterminal.
%For instance, if $G$ has the productions $X_0 \to g(X_1, X_2) %\mid h(X_2, X_3)$,
% the dependence graph has three edges into node $X_0$:
%$X_1{\rightarrow}X_0$, $X_2{\rightarrow}X_0$, and %$X_3{\rightarrow}X_0$.
%To find an order in which to solve the equations:
%1) find the SCCs of the dependence graph,
%2) collapse each SCC into a single node, to form a directed acyclic graph (DAG).
%The equation solver can work through nonterminals 
%in each SCC-node (strata) following any topological order.

\begin{changebar}

The $n_\calG$ equations (\eqref{GFA}) that arise in a GFA problem are amenable to the
standard optimization technique of identifying ``strata'' of dependences among
nonterminals, and solving the equations by finding values for
nonterminals of lower ``strata'' first, working up to higher strata in
an order that respects dependences among the equations.

This idea can be formalized in terms of the strongly connected
components (SCCs) of a dependence graph, defined as follows: the nodes
are the nonterminals of $G$;
the edges represent the dependence of a left-hand-side nonterminal on a right-hand-side
nonterminal.
For instance, if $G$ has the productions $X_0 \to g(X_1, X_2) \mid h(X_2, X_3)$,
then the dependence graph has three edges into node $X_0$:
$X_1 \rightarrow X_0$, $X_2 \rightarrow X_0$, and $X_3 \rightarrow X_0$.
There are three steps to finding an order in which to solve the equations:
\begin{itemize}
	\item
	Find the SCCs of the dependence graph.
	\item
	Collapse each SCC into a single node, to form a directed acyclic graph (DAG).
	\item
	Find a topological order of the DAG.
\end{itemize}
The set of nonterminals associated with a given node of the DAG
corresponds to one of the strata referred to earlier.
The equation solver can work through the strata in any topological order of the DAG.
\end{changebar}

%%% Local Variables: 
%%% mode: latex
%%% TeX-master: "main.tex"
%%% End: 

% -*- TeX-master: t; TeX-PDF-mode: t -*-

\section{Evaluation}
\label{Se:evaluation}

In this section, we evaluate the effectiveness
and performance of \namesl and \namehorn.\footnote{All the experiments were performed on an Intel Core i7 4.00GHz CPU, with 32GB of RAM.
We used version 1.8 of CVC4 and
commit d37c50e of ESolver.
The timeout for each individual \name/ESolver call
is set at 10 minutes.}
\iffull 
First, we present the set of benchmarks we adopt in our experiments.
Second, we evaluate how \name compares to the state-of-the-art tool \nope (\sectref{nopevsnay}).
Third, we evaluate how the performance of \name is affected by
the number of examples required to prove unrealizability 
and by the number of nonterminals in the input grammar (\sectref{timevsexamples}).
Last, we evaluate the effectiveness of the stratification technique presented in \sectref{implementation} (\sectref{optimization-effect})
\fi

\paragraph{Benchmarks.}
We perform our evaluation using \numBenchmarks variants of the 
60 CLIA benchmarks from the CLIA \sygus competition track~\cite{alur2016sygus}.
These benchmarks are the same ones used in the  evaluation of the tool we compare against, \nope~\cite{cav19}, which
like \name only supports LIA and CLIA \sygus problems.

The benchmarks are divided into three categories, and arise from a
tool used to synthesize terms in which a certain syntactic feature
appears a minimal number of times~\cite{HuD18}.
%\begin{enumerate}
%	\item 
\textsc{LimitedPlus} (resp.\ \textsc{LimitedIf}) contains \numPlusBenchmarks
(resp.\ \numIfBenchmarks) benchmarks in which the grammar bounds the number
of times  a \textrm{Plus} (resp.\ \textrm{IfThenElse}) operator can appear
in an expression-tree to be one less than the number
required to solve the original synthesis problem.
\textsc{LimitedConst} contains \numConstBenchmarks benchmarks that
restrict what constants appear in the grammar.
\iffull
The numbers of benchmarks in the three suites differ because for certain benchmarks
it did not make sense to create a limited variant---e.g., if the optimal term
consistent with the specification contains no \textrm{IfThenElse} operators,
no variant is created for the \textsc{LimitedIf} benchmark.
\fi
In each of the benchmarks, the grammar that specifies the search space
generates infinitely many terms.

\subsection{Effectiveness of \name}
\label{Se:nopevsnay}

\begin{resq}%{ProveUnrealizability}
	How effective is \name at proving unrealizability? %(\sectref{EQ1})
\end{resq}

\setlength{\tabcolsep}{4pt}

	\begin{table}[t!]
	\caption{
		Performance of \name and \nope for \textsc{LimitedIf} and \textsc{LimitedPlus} benchmarks.\protect\footnotemark 
		\ The table shows the number of nonterminals ($|N|$), productions ($|\delta|$), and variables ($|V|$)
		in the problem grammar; the number of examples required to prove unrealizability ($|E|$);
		%the number of iterations required to find a fixed point for the IfThenElse guards  (\textbf{ITE iters.}),
		%the total number of Boolean vectors in the final abstraction of Boolean nonterminals ({ $|bset|$}),
		%the size of the final semi-linear set (\textbf{SL size}),
		%how many of the 5 random input examples \name terminated on (\textbf{\name res}),
		%the \textbf{total} time spent by \name as well as the time spent solving integer equations (\textbf{SL})
		%and checking for satisfiability when solving Boolean equations and checking for correctness against
		%the specification (\textbf{SMT}), 
		and the average running time of \namesl, \namehorn, and \nope.
		\xmark\ denotes a timeout.
		}\vspace{-2mm}
		\footnotesize
		\centering
		\begin{tabular}{cc|rrr|r|rr|r}
		&	 \multirow{2}{*}{\bf Problem} & \multicolumn{3}{c|}{\bf Grammar}& \multirow{2}{*}{\bf $|E|$} &  \multicolumn{3}{c}{\bf time (s)}   \\
			%& \multirow{2}{*}{\bf extra col} & \multirow{2}{*}{\bf extra col} \\
			%multicolumn{1}{C{10mm}}{\bf Time} &  \multicolumn{2}{|c}{\bf Time single line} [sec]\\
			 &   &  {\bf $|N|$}  &  {\bf $|\delta|$} & {\bf $|V|$} & &\multirow{1}{*}{\bf \namesl} & \multirow{1}{*}{\bf \namehorn} &\multirow{1}{*}{\bf \nope}  \\
			    \hline
\parbox[t]{-11mm}{\multirow{11}{*}{\rotatebox[origin=c]{90}{\textsc{LimitedPlus}}}}
& guard1	&	7	&	24	&	3	&	2	&	0.24	&	\xmark 	& 	\xmark 	\\
& guard2	&	9	&	34	&	3	&	3	&	12.86	&	\xmark 	&	\xmark 	\\
& guard3	&	11	&	41	&	3	&	1	&	0.07	&	\xmark 	&	\xmark 	\\
& guard4*	&	11 	& 	72 	& 	3 	& 	3.5 	&  147.50	&	\xmark 	&	\xmark 	\\
& plane1	&	2	&	5	&	2	&	1	&	0.07	&	0.55 	& 	0.69\\
& plane2	&	17	&	60	&	2	&	1.6	&	0.90	&	\xmark	&	\xmark\\
& plane3	&	29	&	122	&	2	&	1.5	&	15.73	&	\xmark	&	\xmark\\
& ite1*   	&	7	&	2	&	3	&	2	&	1.05	&	\xmark	&	\xmark\\
& ite2*      & 	9	& 	34 	& 	3 	& 	4 	&   294.88 	&	\xmark	&	\xmark\\
&  sum\_2\_5& 	11 	& 	40 	& 	2 	& 	4 	&	15.48 	& 	\xmark	&	\xmark	\\
&  search\_2&	5	&	16	&	3	&	3	&	1.21	&	\xmark	&	\xmark\\
&  search\_3& 	7 	& 	25 	& 	4	& 	4 	&	2.65	&	\xmark	&	\xmark\\
\hdashline	
\parbox[t]{1mm}{\multirow{13}{*}{\rotatebox[origin=c]{90}{\textsc{LimitedIf}}}}
&max2	&	1	&	5	&	2	&	4	&	0.13	&	1.13	&	1.48 \\
&max3	&	3	&	15	&	3	&	-	&	\xmark 	&	9.67	&	58.57\\
&  sum\_2\_5	&	1	&	5	&	2	&	3	&	0.17	&	0.61	&	0.69\\
&  sum\_2\_15	&	1	&	5	&	2	&	3	&	0.17	& 	0.56	&	0.87\\
&  sum\_3\_5	&	3	&	15	&	3	&	-	&	\xmark	&	17.85	&	101.44 \\
&  sum\_3\_15	&	3	&	15	&	3	&	-	&	\xmark	&	16.65	&	134.87\\
&  search\_2	&	3	&	15	&	3	&	-	&	\xmark	&	25.85	&	112.78\\
& example1	&	3	&	10	&	2	&	3	&	0.14	&	0.73	&	1.12\\
& guard1	&	1	&	6	&	2	&	4	&	0.13	&	0.44	&	0.43\\
& guard2	&	1	&	6	&	2	&	4	&	0.22	& 	0.33	&	0.49\\
& guard3	&	1	&	6	&	2	&	4	&	0.16	&	0.27	&	0.46\\
& guard4	&	1	&	6	&	2	&	4	&	0.11	&	0.72	&	0.58\\
& ite1	&	3	&	15	&	3	&	-	&	\xmark	&	2.68	&	369.57\\
%			 ite1 &	8  & 32 & 3 & - & - & - & - & - & \xmark & \xmark & \xmark &  \xmark &  \xmark\\
%			 ite2 & 10 & 42 & 3 & - & - & - & - & - & \xmark & \xmark & \xmark &  \xmark &  \xmark\\
			
%			 guard4 & 14 & 60 & 3 &  - & - & - & - & - & \xmark & \xmark & \xmark &  \xmark &  \xmark\\
%			 example1 & 52 & 336 & 2 &- & - & - & - & - & \xmark & \xmark & \xmark &  \xmark &  \xmark\\
%			  sum\_2\_5 & 12 & 47 & 2 &- & - & - & - & - & \xmark & \xmark & \xmark &  \xmark &  \xmark\\

\end{tabular}
\label{Ta:results}
\end{table}
\vspace{-2mm}
We compare \namesl and \namehorn against \nope, the state-of-the-art 
tool for proving unrealizability of \sygus problems~\cite{cav19}. 
For each benchmark, we run each tool 5 times on different random
seeds, therefore generating different random sets of examples, and
report whether a tool successfully terminated on at least one run.
This process guarantees that all tools are evaluated on the same
final example set that causes a problem to be unrealizable.
Table~\ref{Ta:results} shows the results for the \textsc{LimitedPlus}
and \textsc{LimitedIf} benchmarks that at least one of the three tools could solve.
Because both tools use a \cegis loop to produce input examples, only
the last iteration of \cegis is unrealizable.
For \namesl and \nope, that iteration is the one that dominates the runtime.
On average, it accounts for 60.4\% of the running time for \namesl and 90.3\% for \nope, but
only 8.3\% for \namehorn.
(For \namehorn, counterexample generation is the most costly step.)
\iffull
Table~\ref{Ta:results2} in \sectref{apptable} shows the detailed result for the \textsc{LimitedConst} benchmarks.
\else
The \textsc{LimitedConst} benchmarks could be solved by all tools, and
results are given in the supplementary material.
\fi
\footnotetext{
	We discovered that three of the benchmarks from~\cite{cav19} were
	actually realizable (marked with *).  Because these benchmarks were created by bounding the number of \textrm{Plus} operators,
	we further reduced the bound by one to make them unrealizable.
}

 \paragraph{Findings.} 
\namesl solved \solved/\numBenchmarks benchmarks,
with an average running time of \avgtime.\iffull\footnote{
  Most of the benchmarks for which  \namesl times out are actually crashes caused by a memory leak in CVC4.
  We have reported the bug.
}\fi
\namehorn and \nope solved identical sets of \solvednope/\numBenchmarks benchmarks,
with an average running time of \avgtimehorn and \avgtimenope, respectively. 
All tools can solve all the \textsc{LimitedConst} benchmarks with similar performance.
These benchmarks are easier than the other ones.

\namesl can solve \naymore \textsc{LimitedPlus} benchmarks that \nope cannot solve. 
These benchmarks involve large grammars, a known weakness of  \nope (see~\cite{cav19}).
\begin{changebar}
	In particular, NaySL can handle grammars with up to 29 nonterminals while Nope can only handle grammars with up to 3 nonterminals.
\end{changebar}
For 8 benchmarks, \namesl only terminated for some of the random runs (certain 
random seeds triggered more CEGIS iterations, making the final problem harder for \name to solve). 

\nope solved 5 \textsc{LimitedIf} benchmarks that \namesl cannot solve.
\nope solves these benchmarks using between 7 and 9 examples in the CEGIS loop.
Because the size of the semi-linear sets computed by \namesl 
depends heavily on the number of examples, \namesl only solves benchmarks that require at most
4 examples. 
\sectref{timevsexamples}  analyzes the effect of the number of examples on 
\namesl's performance.	
When \namesl terminated, it took 1 to 15 iterations (avg.\ 6.6) to find a fixed point for \textrm{IfThenElse} guards, 
and the final abstract domain of each guard contained 2 to 16 Boolean vectors (avg.\ 5.9).	
On average, the running time for computing semi-linear sets is 70.6\% of the total running time.
On the benchmarks that all tools solved, all tools terminated in less than 2s.
	
\namehorn and \nope solved exactly the same set of benchmarks.
This outcome is not surprising because \nope uses \seahorn,
a verification solver based on Horn clauses that builds on Spacer,
which is the constrained-Horn-clause solver used by \namehorn.
\namehorn directly encodes the equation-solving problem, while \nope
reduces the unrealizability problem to a verification problem that is
then translated into a potentially complex constrained-Horn-clause
problem.
For this reason, \namehorn is on average 19 times faster than \nope. 
On benchmarks for which \nope took more than 2 seconds,
\namehorn is 82x faster than \nope (computed as the geometric mean).
\begin{changebar}
	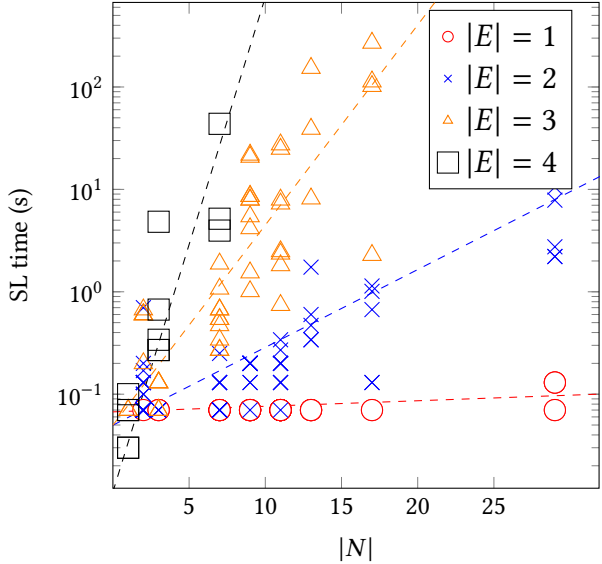
\begin{figure}[t]
	\centering
	\begin{tikzpicture}[every mark/.append style={mark size=4pt}]
	\begin{axis}[%
	legend style={nodes={scale=1.25, transform shape},at={(0.65,0.8)},anchor=west},
	ymode = log,
	ylabel absolute, ylabel style={yshift=0mm},
	xlabel absolute, xlabel style={yshift=0mm},
	xtick ={5,10,15,20,25},
	width=0.95\linewidth,
	height=0.95\linewidth,
	scatter/classes={%
		a={mark=*,draw=red,
			mark options={solid},
			style={solid, fill=white}},
		b={mark=x,draw=blue,
			mark options={solid},
			style={solid, fill=white}},
		c={mark=triangle,draw=orange,
			mark options={solid},
			style={solid, fill=white}},
		d={mark=square,draw=black,
			style={solid, fill=black}}},
	xmin = 0,
	ymin = 0,
	xlabel={$|N|$},
	ylabel={SL time (s)}]
	\addplot[scatter,only marks,%
	scatter src=explicit symbolic]%\
	table[meta=label] {
		x y label
2	0.07	a
11	0.07	a
3	0.07	a
3	0.07	a
3	0.07	a
17	0.07	a
17	0.07	a
29	0.07	a
29	0.07	a
29	0.13	a
29	0.13	a
29	0.13	a
7	0.07	a
7	0.07	a
9	0.07	a
9	0.07	a
9	0.07	a
9	0.07	a
7	0.07	a
7	0.07	a
7	0.07	a
7	0.07	a
9	0.07	a
9	0.07	a
9	0.07	a
9	0.07	a
9	0.07	a
11	0.07	a
11	0.07	a
11	0.07	a
11	0.07	a
11	0.07	a
13	0.07	a
13	0.07	a
13	0.07	a
13	0.07	a
13	0.07	a
11	0.07	a
11	0.07	a
11	0.07	a
11	0.07	a
11	0.07	a
2	0.07	b
2	0.07	b
2	0.07	b
2	0.07	b
2	0.07	b
2	0.07	b
2	0.1	b
2	0.1	b
2	0.13	b
2	0.13	b
2	0.13	b
2	0.17	b
2	0.2	b
2	0.13	b
2	0.07	b
2	0.07	b
2	0.07	b
2	0.07	b
2	0.07	b
2	0.07	b
2	0.07	b
2	0.07	b
2	0.7	b
2	0.07	b
2	0.07	b
2	0.07	b
2	0.07	b
2	0.07	b
2	0.07	b
7	0.25	b
29	2.21	b
3	0.07	b
3	0.07	b
3	0.07	b
17	1.01	b
17	0.67	b
17	0.13	b
17	0.13	b
17	1.14	b
29	19.46	b
29	10.54	b
29	7.85	b
29	2.75	b
29	2.21	b
7	0.13	b
7	0.13	b
7	0.07	b
7	0.07	b
7	0.13	b
9	0.2	b
9	0.07	b
9	0.2	b
9	0.2	b
9	0.13	b
7	0.07	b
7	0.07	b
7	0.07	b
7	0.13	b
7	0.07	b
9	0.2	b
9	0.13	b
9	0.2	b
9	0.2	b
9	0.2	b
11	0.2	b
11	0.2	b
11	0.34	b
11	0.2	b
11	0.27	b
13	0.34	b
13	0.34	b
13	1.74	b
13	0.47	b
13	0.6	b
11	0.13	b
11	0.13	b
11	0.2	b
11	0.13	b
11	0.07	b
2	0.2	c
2	0.6	c
2	0.67	c
2	0.6	c
2	0.6	c
1	0.07	c
1	0.07	c
9	7.92	c
3	0.13	c
3	0.07	c
3	0.07	c
3	0.13	c
3	0.13	c
17	112.81	c
17	101.87	c
17	2.28	c
17	269.37	c
7	0.34	c
7	0.54	c
7	0.54	c
7	0.67	c
7	1.88	c
9	8.59	c
9	1.01	c
9	8.59	c
9	20.74	c
7	0.27	c
7	0.27	c
7	0.47	c
7	1.07	c
7	0.67	c
9	5.44	c
9	4.16	c
9	21.74	c
9	1.54	c
9	7.92	c
11	7.25	c
11	24.83	c
11	27.38	c
11	7.92	c
13	38.92	c
13	153.01	c
13	8.12	c
11	1.81	c
11	2.48	c
11	2.35	c
11	0.74	c
1	0.07	d
1	0.03	d
1	0.03	d
1	0.1	d
1	0.07	d
3	0.27	d
3	0.27	d
3	0.67	d
3	4.83	d
3	0.34	d
7	43.69	d
7	3.96	d
7	5.17	d
	};
	\legend{$|E|=1$,$|E|=2$,$|E|=3$,$|E|=4$}
	%\addplot+[line width=1pt,mark size=.5pt] coordinates {(0,0) (300000,300000)};
	%\addplot[color=red] coordinates { (1,0) (3,1) (6, 3) (9, 9) (10,15) (11,25) (12,40) (13, 60) (14,85) (16,140) (20,250)(21,280) (22,320)};
	%\node [below, xshift=-1ex] at (axis cs:25.69,500) {\small TO};
	\draw [dashed, color=red] (0,-2.7) -- (400,-2.2);
	\draw [dashed, color=blue] (0,-3) -- (400,4);
	\draw [dashed, color=orange] (0,-3) -- (400,15);
	\draw [dashed, color=black] (0,-4.5) -- (400,40);
	\end{axis}
	\end{tikzpicture}	\vspace{-3mm}
	\caption{Time to compute semi-linear set vs.\ $|N|$.\label{Fig:complexity}}	
	\vspace{-2mm}
\end{figure}
	The reason we use random examples in \algref{CEGIS} is that there is a trade-off between the size of solutions and the number of examples when we are proving the realizability of \sygus-with-examples problems. On the one hand, ESolver is not affected by the number of examples, and can efficiently synthesize a solution when a small solution exists.
	On the other hand the time required to prove realizability by \namesl only depends on the size of grammars and the number of examples but not on the size of solutions.
For the realizable \sygus-with-examples problems produced during the CEGIS loop of our experiments, ESolver terminates on average in 1.9 seconds when there exists a solution with size no more than 10, but terminates on average in 54.5 seconds when there exists a solution with size greater than 10 (the largest solution has size 24). For the same problems, \namesl could not prove realizability for problems with more than 5 examples, but it did prove realizability for 7 problems on which ESolver failed. On the problems both ESolver and \namesl solved, ESolver is 87\%  faster than \namesl calculated as a geometric mean.
\end{changebar}

To answer \textbf{EQ 1}: if both \name techniques are considered together,
\emph{\name solved \naymore benchmarks that \nope did not solve, and was faster
on the benchmarks that both tools solved.}

\subsection{The Cost of Proving Unrealizability}
\label{Se:timevsexamples}
\begin{resq}
	How does the size of the grammar and the number of examples affect the 
	performance of different solvers? %(\sectref{EQ2})
\end{resq}
\vspace{-1mm}
%\iffull
%The size of the  semi-linear sets computed by 
%\namesl can be double	 exponential in the number $n$ of nonterminals 
%in the grammar and the size $k$ of the Boolean abstraction set $bset$, which is in the worst case bounded by
%$2^{|E|}$.
%We analyze whether this worst case analysis happens in practice.
%\fi
%Figure~\ref{Fig:complexity} plots the time taken to compute the  semi-linear set against the number of nonterminals (y-axis log-scale).

\paragraph{Finding.}
\begin{changebar}
First, consider \namesl: when we fix the number of examples (different marks in Fig.~\ref{Fig:complexity}),
the time taken to compute the semi-linear set
grows roughly exponentially.
%(instead of the theoretical doubly-exponential worst-case) with  $|N|$.
Also, the time grows roughly exponentially 
%(instead of double-exponentially)
 with respect to $2^{|E|}$.
 	
 	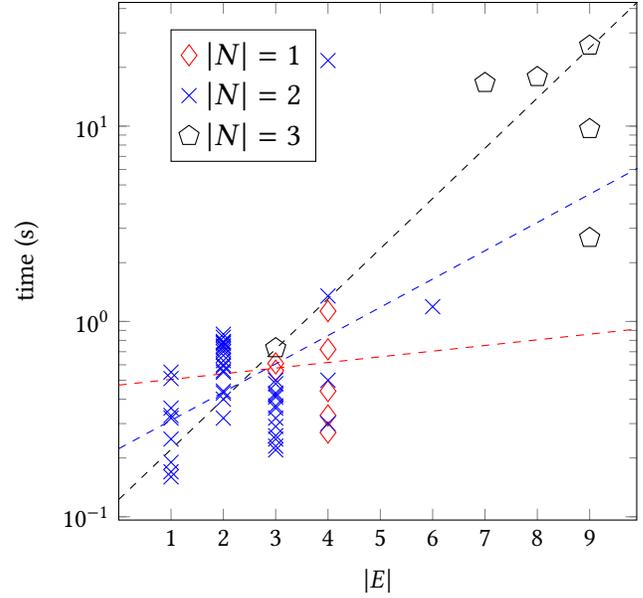
\begin{figure}[t!]

	\begin{tikzpicture}[>=latex]
	\begin{axis}[%
	legend style={nodes={scale=1.2, transform shape},at={(0.1,0.82)},anchor=west},
	ymode = log,
	ylabel absolute, ylabel style={yshift=0},
	xlabel absolute, xlabel style={yshift=0},
	xtick ={1,2,3,4,5,6,7,8,9,10},
	width=\linewidth,
	height=\linewidth,
	scatter/classes={%
		a={mark=diamond,draw=red,mark size=4pt,
			mark options={solid},
			style={solid, fill=white}},
		b={mark=x,draw=blue,mark size=4pt,
			mark options={solid},
			style={solid, fill=white}},
		c={mark=pentagon,draw=black,mark size=4pt,
			mark options={solid},
			style={solid, fill=white}}},
	xmin = 0,
	ymin = 0,
	xlabel={$|E|$},
	ylabel={time (s)}]
	\addplot[scatter,only marks,%
	scatter src=explicit symbolic]%\
	table[meta=label] {
		x y label
				1	0.55	b
		4	1.13	a
		9	9.67	c
		3	0.61	a
		3	0.56	a
		8	17.85	c
		7	16.65	c
		9	25.85	c
		3	0.73	c
		4	0.44	a
		4	0.33	a
		4	0.27	a
		4	0.72	a
		9	2.68	c
		2	0.32	b
		3	0.43	b
		3	0.22	b
		3	0.5	b
		3	0.41	b
		3	0.26	b
		3	0.36	b
		3	0.48	b
		3	0.37	b
		3	0.32	b
		3	0.25	b
		3	0.29	b
		3	0.42	b
		3	0.23	b
		2	0.77	b
		2	0.83	b
		2	0.86	b
		2	0.75	b
		2	0.6	b
		2	0.56	b
		2	0.72	b
		2	0.44	b
		2	0.55	b
		2	0.7	b
		2	0.59	b
		2	0.78	b
		2	0.63	b
		2	0.67	b
		2	0.4	b
		2	0.79	b
		2	0.74	b
		2	0.43	b
		1	0.17	b
		4	0.3	b
		1	0.33	b
		1	0.16	b
		1	0.36	b
		6	1.19	b
		4	1.35	b
		4	0.5	b
		4	21.67	b
		1	0.19	b
		1	0.25	b
		1	0.32	b
		1	0.51	b
	};
	\legend{$|N|=1$,$|N|=2$,$|N|=3$}
	%\addplot+[line width=1pt,mark size=.5pt] coordinates {(0,0) (300000,300000)};
	%\addplot[color=red] coordinates { (1,0) (3,1) (6, 3) (9, 9) (10,15) (11,25) (12,40) (13, 60) (14,85) (16,140) (20,250)(21,280) (22,320)};
	%\node [below, xshift=-1ex] at (axis cs:25.69,500) {\small TO};
	\draw [dashed, color=red] (0,-0.75) -- (1200,0.05);
	\draw [dashed, color=blue] (0,-1.5) -- (1200,2.5);
	\draw [dashed, color=black] (0,-2.1) -- (1200,5);
	\end{axis}
	\end{tikzpicture}
	\caption{Running time of \namehorn vs. number of examples.}
	\label{fig:nayhorn}
	\vspace{-1mm}
\end{figure}%
 	\namehorn and \nope (shown in Fig.~\ref{fig:nayhorn} and Fig.~\ref{fig:nope}, respectively) can only solve benchmarks involving up to 3 nonterminals.  
 	When we fix the number of nonterminals, the running time of these two tools grows roughly  exponentially with respect to the number of examples.

To answer \textbf{EQ 2}: the running time of \namesl grows exponentially with respect to $|N|2^{|E|}$, and the running time of \namehorn and \nope grows exponentially with respect to $|E|$.
 \end{changebar}

\subsection{Effectiveness of Grammar Stratification}
\label{Se:optimization-effect}
\begin{resq}%{HowManyExamplesNeeded}
	Is the stratification optimization from \sectref{implementation} effective? %(\sectref{EQ2})
\end{resq}
\vspace{-1mm}

\iffull
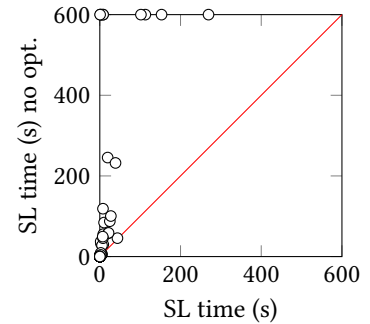
\begin{wrapfigure}{r}{0.3\textwidth}
	\centering
	\vspace{-0.4cm}	
	\begin{tikzpicture}[scale=1,every mark/.append style={mark size=2pt}]
	\begin{axis}[%
	ylabel absolute, ylabel style={yshift=-3mm},
	xlabel absolute, xlabel style={yshift=1mm},
	xtick =data,
	width=0.9\linewidth,
	height=0.9\linewidth,
	scatter/classes={%
		a={mark=*,draw=black,
			mark options={solid},
			style={solid, fill=white}},
		b={mark=x,draw=black,
			mark options={solid, mark size=2pt},
			style={solid, fill=white}, mark size =3pt},
		c={mark=square*,draw=black,
			mark options={solid, mark size=2pt},
			style={solid, fill=white}, mark size =2pt},
		d={mark=triangle,draw=black,
			style={solid, fill=black}}},
	xmin = 0,
	ymin = 0,
	xmax = 600,
	ymax = 600,
	xlabel={SL time (s)},
	ylabel={SL time (s) no opt.},
	xtick = {0,200,400,600}
	]
	\addplot[scatter,only marks,%
	scatter src=explicit symbolic]%
	table[meta=label] {
		x y label
0.07	0	a
0.07	0.07	a
0.13	0.13	a
0.27	0.27	a
0.07	0	a
0.07	0.07	a
0.07	0.07	a
0.27	0.27	a
0.01	0.07	a
0.01	0.07	a
0.07	0.07	a
0.67	0.6	a
0.07	0	a
0.07	0.07	a
0.13	0.07	a
4.83	4.83	a
0.07	0	a
0	0.07	a
0.13	0.13	a
0.34	0.27	a
0.07	0.07	a
1.01	4.83	a
112.81	600	a
0.07	0.07	a
0.67	2.21	a
101.87	600	a
0.07	0.07	a
0.13	0.74	a
2.28	600	a
0.07	0.07	a
0.13	0.54	a
0.07	0.13	a
1.14	5.77	a
269.37	600	a
0.07	0.27	a
19.46	245.55	a
0.07	0.2	a
10.54	84.36	a
0.13	0.2	a
7.85	118.72	a
0.13	0.2	a
2.75	32.08	a
0.13	0.2	a
2.21	37.38	a
0.07	0.07	a
0.13	0.07	a
0.34	0.87	a
0	0.07	a
0.13	0.13	a
0.54	1.48	a
0	0.07	a
0.07	0.07	a
0.54	0.81	a
0.07	0.07	a
0.07	0.13	a
0.67	4.7	a
0	0	a
0.13	0.07	a
1.88	1.21	a
43.69	45.7	a
0.07	0.07	a
0.2	0.4	a
8.59	29.93	a
0.07	0.07	a
0.07	0.13	a
1.01	1.21	a
0	0.07	a
0.2	0.4	a
8.59	55.57	a
0.07	0.07	a
0.2	0.4	a
20.74	57.85	a
0.07	0.07	a
0.13	0.27	a
0.07	0.07	a
0.07	0.07	a
0.27	0.81	a
0.07	0.07	a
0.07	0.07	a
0.27	1.07	a
0.07	0.07	a
0.07	0.07	a
0.47	0.54	a
3.96	4.97	a
0.07	0.07	a
0.13	0.07	a
1.07	0.6	a
5.17	6.91	a
0.01	0.07	a
0.07	0.13	a
0.67	1.41	a
0.07	0.13	a
0.2	0.27	a
5.44	27.78	a
0.07	0.13	a
0.13	0.2	a
4.16	5.1	a
0.07	0.13	a
0.2	0.47	a
21.74	58.72	a
0.07	0.13	a
0.2	0.47	a
1.54	600	a
0.07	0.07	a
0.2	0.4	a
7.92	45.63	a
0.07	0.07	a
0.2	0.4	a
7.25	49.53	a
0.07	0.07	a
0.2	0.6	a
24.83	87.64	a
0.07	0.07	a
0.34	0.74	a
27.38	100.53	a
0.07	0.07	a
0.2	0.13	a
7.92	600	a
0.07	0.07	a
0.27	0.4	a
0.07	0.07	a
0.34	0.54	a
38.92	232.2	a
0.07	0.07	a
0.34	0.81	a
0.07	0.13	a
1.74	1.41	a
153.01	600	a
0.07	0.07	a
0.47	1.61	a
8.12	600	a
0.07	0.13	a
0.6	1.41	a
0.07	0.07	a
0.13	0.13	a
1.81	4.36	a
0.07	0.07	a
0.13	0.13	a
2.48	9.13	a
0.07	0.13	a
0.2	0.27	a
2.35	4.09	a
0.07	0.07	a
0.13	0.13	a
0.07	0.07	a
0.07	0.13	a
0.74	600	a
	};
%	\legend{$|E|=1$,$|E|=2$,$|E|=3$,$|E|=4$}
	%\addplot+[line width=1pt,mark size=.5pt] coordinates {(0,0) (300000,300000)};
	%\addplot[color=red] coordinates { (1,0) (3,1) (6, 3) (9, 9) (10,15) (11,25) (12,40) (13, 60) (14,85) (16,140) (20,250)(21,280) (22,320)};
	%\node [below, xshift=-1ex] at (axis cs:25.69,500) {\small TO};
	\draw [solid, color=red] (0,0) -- (600,600);
	\end{axis}
	\end{tikzpicture}
	\vspace{-0.4cm}
	\caption{Stratification speedup.	\label{Fi:optimization}}
\end{wrapfigure}
Figure~\ref{Fi:optimization} shows the time taken to compute the final semi-linear with and without the 
 stratification technique from \sectref{implementation}. Every point above the diagonal line
is a benchmark for which \namesl performs better with the optimization.
\fi

\paragraph{Finding.}
Using stratification, \namesl can compute the semi-linear sets for 9 benchmarks 
for which \namesl times out without the optimization.
On benchmarks that take more than 1s to solve,
the optimization results on average in a 3.1x speedup.
To answer \textbf{EQ 3}: the grammar-stratification optimization is highly effective.

%%% Local Variables: 
%%% mode: latex
%%% TeX-master: "main.tex"
%%% End: 

% -*- TeX-master: t; TeX-PDF-mode: t -*-

\section{Related Work}
\label{Se:RelatedWork}
\mypar{Unrealizability in \sygus.}
%While proving whether
%a given \sygus problem is realizable is an undecidable problem~\cite{CaulfieldRST15},
Several \sygus solvers 
compete in yearly \sygus competitions~\cite{alur2016sygus}, and can
produce solutions to \sygus problems when a solution exists.
If the problem is unrealizable, these solvers only terminate if the
language of the grammar is finite or contains finitely many
functionally distinct programs, which is not the case in our
benchmarks.

\nope~\cite{cav19}, the tool we compare against in
\sectref{evaluation}, is the only tool that can prove unrealizability
for non-trivial \sygus problems.
\nope reduces the problem of proving unrealizability to one of
proving unreachability in a recursive non-deterministic program, and
uses off-the-shelf verifiers to solve the unreachability problem.
%As we showed \name and \nope are quite comparable.
Unlike \name, \nope does not provide any insights into how we can
devise specialized techniques for solving unrealizability, because
\nope reduces a constrained \sygus problem to a full-fledged
program-reachability problem.
In contrast, the approach presented in this paper gives a
characterization of unrealizability in terms of solving a set of
equations.
Using the equation-solving framework, we provided the first
\textit{decision procedures} for LIA and CLIA \sygus problems over
examples.
Moreover, the equation-based approach allows us to use known
equation-solving techniques, such as Newton's method and constrained
Horn clauses.
\begin{figure}[t]
	\begin{tikzpicture}[>=latex]
	\begin{axis}[%
	legend style={nodes={scale=1.2, transform shape},at={(0.1,0.82)},anchor=west},
	ymode = log,
	ylabel absolute, ylabel style={yshift=0},
	xlabel absolute, xlabel style={yshift=0},
	xtick ={1,2,3,4,5,6,7,8,9,10},
	width=1\linewidth,
	height=\linewidth,
	scatter/classes={%
	a={mark=diamond,draw=red,mark size=4pt,
		mark options={solid},
		style={solid, fill=white}},
	b={mark=x,draw=blue,mark size=4pt,
		mark options={solid},
		style={solid, fill=white}},
	c={mark=pentagon,draw=black,mark size=4pt,
		mark options={solid},
		style={solid, fill=white}}},
	xmin = 0,
	ymin = 0,
	xlabel={$|E|$},
	ylabel={time (s)}]
	\addplot[scatter,only marks,%
	scatter src=explicit symbolic]%\
	table[meta=label] {
		x y label
		
		1	0.69	b
		4	1.48	a
		9	58.5	c
		3	0.69	a
		3	0.87	a
		8	101.44 	c
		7	134.87	c
		9	112.78	c
		3	1.12	c
		4	0.43	a
		4	0.49	a
		4	0.46	a
		4	0.58	a
		9	369.57	c
		2	0.78	b
		3	1.26	b
		3	1.25	b
		3	1.01	b
		3	0.87	b
		3	0.85	b
		3	0.97	b
		3	0.7	b
		3	0.8	b
		3	1.09	b
		3	1.13	b
		3	0.73	b
		3	0.77	b
		3	1.06	b
		2	1.3	b
		2	1.46	b
		2	1.31	b
		2	1.28	b
		2	2.52	b
		2	1.35	b
		2	1.41	b
		2	1.43	b
		2	2.37	b
		2	1.56	b
		2	0.76	b
		2	1.87	b
		2	1.33	b
		2	1.53	b
		2	1.5	b
		2	1.44	b
		2	2.29	b
		2	0.87	b
		1	0.36	b
		4	0.5	b
		1	0.57	b
		1	0.44	b
		1	0.99	b
		6	3.08	b
		4	2.49	b
		4	1.83	b
		4	24.18	b
		1	0.33	b
		1	0.41	b
		1	0.47	b
		1	0.74	b
	};
	\legend{$|N|=1$,$|N|=2$,$|N|=3$}
	%\addplot+[line width=1pt,mark size=.5pt] coordinates {(0,0) (300000,300000)};
	%\addplot[color=red] coordinates { (1,0) (3,1) (6, 3) (9, 9) (10,15) (11,25) (12,40) (13, 60) (14,85) (16,140) (20,250)(21,280) (22,320)};
	%\node [below, xshift=-1ex] at (axis cs:25.69,500) {\small TO};
	
	\draw [dashed, color=red] (0,-0.35) -- (1200,0.05);
	\draw [dashed, color=blue] (0,-1) -- (1200,4.8);
	\draw [dashed, color=black] (0,4) -- (1200,5.5);
	%	\draw [dashed, color=black] (0,-4.5) -- (400,40);
	\end{axis}
	\end{tikzpicture}
	\caption{Running time of \nope vs. number of examples.}
	\label{fig:nope}
	\vspace{-1mm}
\end{figure}
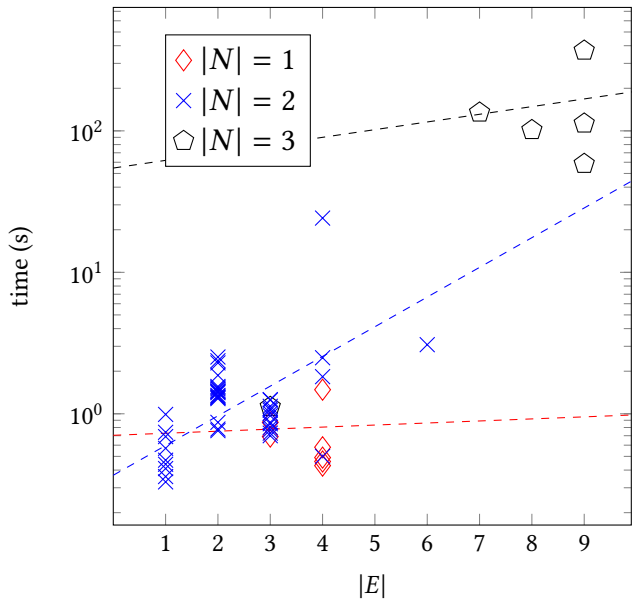

\mypar{Unrealizability in Program Synthesis.}
For certain synthesis problems---e.g., reactive
synthesis~\cite{Bloem15}---realizability is decidable.
However, \sygus is orthogonal to such problems.

Mechtaev et al.~\cite{Mechtaev18} propose to use unrealizability to   prune irrelevant paths
 in symbolic-execution engines.
%For each path $\pi$ in the program, a synthesis
%problem $p_\pi$ is generated so that if $p_\pi$ is unrealizable, the
%path $\pi$ is infeasible.
The synthesis problems generated by Mechtaev et al. are not directly expressible
in \sygus. Moreover, these problems are decidable
because they can be encoded as SMT formulas.

\mypar{Abstractions in Program Synthesis.}
SYNGAR~\cite{WangDS18} uses predicate abstraction to prune the search space of a  
synthesis-from-examples problem.
Given an input example $i$ and a regular-tree grammar $A$ representing the search space,
SYNGAR builds a new grammar $A_\alpha$ in which each nonterminal is a pair $(q,a)$,
where $q$ is a nonterminal of $A$ and $a$ is a predicate of a
predicate-abstraction domain $\alpha$.
Any term that can be derived from $(q,a)$ is guaranteed to produce 
an output satisfying the predicate $a$ when fed the input $i$. 
$A_\alpha$ is constructed iteratively by adding  nonterminals in
a bottom-up fashion;
it is guaranteed to terminate because the set $\alpha$ is finite.
SYNGAR can be viewed as a special case of our framework in which   
the set of values $n_\calG(X)$ is based on predicate abstraction (see \sectref{algorithm}).
%When $n_\calG(X)$ is finite, a fixed-point algorithm based on iterations  is always guaranteed to terminate, which is why SYNGAR employs a simple bottom-up
%approach. SYNGAR's bottom-up approach cannot terminate when the domain of $\alpha$ is infinite 
%because the final tree automaton would require infinitely many states. 
SYNGAR's approach is tied to finite abstract domains,
while our equational approach extends to infinite domains---e.g., semi-linear sets---because
it does not specify how the equations must be solved.

%%% Local Variables: 
%%% mode: latex
%%% TeX-master: "main.tex"
%%% End: 

%\input{conclusion}

%% Acknowledgments
%% Acknowledgments

\begin{acks}                            %% acks environment is optional
	
	%% contents suppressed with 'anonymous'
	
	%% Commands \grantsponsor{<sponsorID>}{<name>}{<url>} and
	
	%% \grantnum[<url>]{<sponsorID>}{<number>} should be used to
	
	%% acknowledge financial support and will be used by metadata
	
	%% extraction tools.
	
	Supported, in part,
	by a gift from Rajiv and Ritu Batra;
	by \grantsponsor{GS100000003}{ONR}{https://www.onr.navy.mil/
	} under grants
~\grantnum{GS100000003}{N00014-17-1-2889} and
~\grantnum{GS100000003}{N00014-19-1-2318}; by \grantsponsor{GS100000003}{NSF}{} under grants ~\grantnum{GS100000003}{1763871} and
~\grantnum{GS100000003}{1750965}; and by a Facebook fellowship.
The U.S.\ Government is authorized to reproduce and distribute
reprints for Governmental purposes notwithstanding any copyright
notation thereon.
Opinions, findings, conclusions, or recommendations
expressed in this publication are those of the authors,
and do not necessarily reflect the views of the sponsoring
agencies.
\end{acks}

%% Bibliography

%% Appendix
\iffull
\appendix
\section{Additional Table}
\label{Se:apptable}
This appendix contains Table~\ref{Ta:results2}, which contains additional statistics on the comparison 
of \name and \nope.

\begin{table}[h]		\caption{
		Performance of \name and \nope for \textsc{LimitedConst} benchmarks.
		See caption of Table~\ref{Ta:results} for the description of the table columns.
	}
	\footnotesize
	\centering
		\begin{tabular}{cc|rrr|r|rr|r}
			&	 \multirow{2}{*}{\bf Problem} & \multicolumn{3}{c|}{\bf Grammar}& \multirow{2}{*}{\bf $|E|$} &\multicolumn{3}{c}{\bf time (s)}   \\
			%& \multirow{2}{*}{\bf extra col} & \multirow{2}{*}{\bf extra col} \\
			%multicolumn{1}{C{10mm}}{\bf Time} &  \multicolumn{2}{|c}{\bf Time single line} [sec]\\
			&   &  {\bf $|N|$}  &  {\bf $|\delta|$} & {\bf $|V|$} & &\multirow{1}{*}{\bf \namesl} & \multirow{1}{*}{\bf \namehorn} &\multirow{1}{*}{\bf \nope}  \\
			\hline
		\parbox[t]{1mm}{\multirow{45}{*}{\rotatebox[origin=c]{90}{\textsc{LimitedConst}}}}&array\_search\_2	&	2	&	10	&	3	&	2	&	0.17	&	0.04	&	0.78\\	
		&array\_search\_3	&	2	&	11	&	4	&	2	&	0.30	&	0.04	&	1.26\\	
		&array\_search\_4	&	2	&	12	&	5	&	2	&	0.47	&	0.01	&	1.25\\	
		&array\_search\_5	&	2	&	13	&	6	&	2	&	0.57	&	0.04	&	1.01\\	
		&array\_search\_6	&	2	&	14	&	7	&	2	&	0.77	&	0.03	&	0.87\\	
		&array\_search\_7	&	2	&	15	&	8	&	2	&	0.97	&	0.03	&	0.85\\	
		&array\_search\_8	&	2	&	16	&	9	&	2	&	1.28	&	0.04	&	0.97\\	
		&array\_search\_9	&	2	&	17	&	10	&	2	&	1.58	&	0.04	&	0.70\\	
		&array\_search\_10	&	2	&	18	&	11	&	2	&	1.88	&	0.04	&	0.80\\	
		&array\_search\_11	&	2	&	19	&	12	&	2	&	2.21	&	0.01	&	1.09\\	
		&array\_search\_12	&	2	&	20	&	13	&	2	&	2.62	&	0.02	&	1.13\\	
		&array\_search\_13	&	2	&	21	&	14	&	2	&	3.05	&	0.05	&	0.73\\	
		&array\_search\_14	&	2	&	22	&	15	&	2	&	3.49	&	0.05	&	0.77\\	
		&array\_search\_15	&	2	&	23	&	16	&	2	&	3.79	&	0.03	&	1.06\\	
		&array\_sum\_2\_5	&	2	&	9	&	2	&	2	&	0.13	&	0.04	&	1.30\\	
		&array\_sum\_2\_15	&	2	&	9	&	2	&	2	&	0.14	&	0.01	&	1.46\\	
		&array\_sum\_3\_5	&	2	&	10	&	3	&	2	&	0.07	&	0.01	&	1.31\\	
		&array\_sum\_3\_15	&	2	&	10	&	3	&	2	&	0.07	&	0.04	&	1.28\\	
		&array\_sum\_4\_5	&	2	&	11	&	4	&	2	&	0.13	&	0.03	&	2.52\\	
		&array\_sum\_4\_15	&	2	&	11	&	4	&	2	&	0.34	&	0.05	&	1.35\\	
		&array\_sum\_5\_5	&	2	&	12	&	5	&	2	&	0.07	&	0.02	&	1.41\\	
		&array\_sum\_5\_15	&	2	&	12	&	5	&	2	&	0.34	&	0.07	&	1.43\\	
		&array\_sum\_6\_5	&	2	&	13	&	6	&	2	&	0.14	&	0.10	&	2.37\\	
		&array\_sum\_6\_15	&	2	&	13	&	6	&	2	&	0.34	&	0.02	&	1.56\\	
		&array\_sum\_7\_5	&	2	&	14	&	7	&	2	&	0.14	&	0.01	&	0.76\\	
		&array\_sum\_7\_15	&	2	&	14	&	7	&	2	&	0.34	&	0.08	&	1.87\\	
		&array\_sum\_8\_5	&	2	&	15	&	8	&	2	&	0.07	&	0.09	&	1.33\\	
		&array\_sum\_8\_15	&	2	&	15	&	8	&	2	&	0.13	&	0.10	&	1.53\\	
		&array\_sum\_9\_5	&	2	&	16	&	9	&	2	&	0.07	&	0.01	&	1.50\\	
		&array\_sum\_9\_15	&	2	&	16	&	9	&	2	&	0.34	&	0.03	&	1.44\\	
		&array\_sum\_10\_5	&	2	&	17	&	10	&	2	&	0.07	&	0.03	&	2.29\\	
		&array\_sum\_10\_15	&	2	&	17	&	10	&	2	&	0.27	&	0.07	&	0.87\\	
		&mpg\_example1	&	2	&	9	&	2	&	1	&	0.07	&	0.05	&	0.36\\	
		&mpg\_example2	&	2	&	9	&	3	&	3	&	5.17	&	0.09	&	0.50\\	
		&mpg\_example3	&	2	&	10	&	3	&	1	&	0.07	&	0.03	&	0.57\\	
		&mpg\_example4	&	2	&	11	&	4	&	1	&	0.07	&	0.04	&	0.44\\	
		&mpg\_example5	&	2	&	9	&	2	&	1	&	0.01	&	0.08	&	0.99\\	
		&mpg\_guard1	&	2	&	10	&	3	&	3	&	15.84	&	0.01	&	3.08	\\
		&mpg\_guard2	&	2	&	10	&	3	&	3	&	16.44	&	0.03	&	2.49	\\
		&mpg\_guard3	&	2	&	10	&	3	&	3	&	15.57	&	0.08	&	0.44	\\
		&mpg\_guard4	&	2	&	10	&	3	&	3	&	15.70	&	1.44	&	24.18	\\
		&mpg\_ite1	&	2	&	10	&	3	&	1	&	0.01	&	0.02	&	0.33\\	
		&mpg\_ite2	&	2	&	10	&	3	&	1	&	0.07	&	0.18	&	0.41\\	
		&mpg\_plane2	&	2	&	10	&	3	&	1	&	0.07	&	0.12	&	0.47\\	
		&mpg\_plane3	&	2	&	10	&	3	&	1	&	0.07	&	0.08	&	0.74	
			\end{tabular}
			\label{Ta:results2}
			\vspace{-2mm}
\end{table}
\fi

\end{document}